\documentclass[aps,pra,onecolumn,letterpaper,superscriptaddress,11pt,floatfix,nofootinbib]{revtex4}
\usepackage{times}
\usepackage{mathtools}
\usepackage{algorithm}
\usepackage{algorithmic}

\usepackage{multirow}
\usepackage[braket,qm]{qcircuit}
\usepackage{amsmath}
\usepackage{amssymb}
\usepackage{bm}
\usepackage{epsfig}
\usepackage{epstopdf}
\usepackage{subfigure}
\usepackage{array}
\usepackage{multirow}
\usepackage{eucal}
\usepackage{makecell}
\usepackage{fnbreak}
\usepackage{amsthm}
\newtheorem{remark}{Remark}
\newtheorem{theorem}{Theorem}
\newtheorem{definition}{Definition}
\newtheorem{proposition}{Proposition}
\newtheorem{lemma}{Lemma}

\usepackage{relsize}
\usepackage{graphicx}
\usepackage{caption}
\usepackage{epsfig}
\usepackage{epstopdf}
\usepackage{subfigure}
\usepackage[qm]{qcircuit}
\usepackage{array}

\DeclareMathOperator{\tr}{tr}

\usepackage{color}
\usepackage{url}
\usepackage{textcomp} 
\usepackage{enumitem}
\setlist{noitemsep,leftmargin=*}
\usepackage[colorlinks,linkcolor=black]{hyperref}
\hypersetup{
colorlinks=true,
urlcolor= blue,
citecolor=blue,
linkcolor= blue,
}

\begin{document}
\title{Quantum algorithms for estimating quantum entropies}
\author{Youle Wang}
\affiliation{Institute for Quantum Computing, Baidu Research, Beijing 100193, China}
\affiliation{Center for Quantum Software and Information, University of Technology Sydney, NSW 2007, Australia}
\author{Benchi Zhao}
\affiliation{Institute for Quantum Computing, Baidu Research, Beijing 100193, China}
\author{Xin Wang}
\affiliation{Institute for Quantum Computing, Baidu Research, Beijing 100193, China}

\begin{abstract}
The von Neumann and quantum R\'enyi entropies characterize fundamental properties of quantum systems and lead to theoretical and practical applications in many fields. Quantum algorithms for estimating quantum entropies, using a quantum query model that prepares the purification of the input state, have been established in the literature. {However, constructing such a model is almost as hard as state tomography.} In this paper, we propose quantum algorithms to estimate the von Neumann and quantum $\alpha$-R\'enyi entropies of an $n$-qubit quantum state $\rho$ using independent copies of the input state. We also show how to efficiently construct the quantum circuits for {quantum entropy estimation} using primitive single/two-qubit gates. We prove that the number of required copies scales polynomially in $1/\epsilon$ and $1/\Lambda$, where $\epsilon$ denotes the additive precision and $\Lambda$ denotes the lower bound on all non-zero eigenvalues. Notably, our method outperforms previous methods in the aspect of practicality since it does not require any quantum query oracles, which are usually necessary for previous methods. Furthermore, we conduct experiments to show the efficacy of our algorithms to single-qubit states and study the noise robustness. We also discuss the applications to some quantum states of practical interest as well as some meaningful tasks such as quantum Gibbs state preparation and entanglement estimation. 
\end{abstract}
\maketitle

\section{Introduction}
Entropy \cite{BEIN2006101} is a vital concept in physics and computer science that can characterize the system's randomness. The celebrated Shannon entropy \cite{shannon2001mathematical} and R\'enyi entropies \cite{renyi1961measures} are often used to depict the randomness and capture the operational quantities in information processing and quantum physics. The R\'enyi entropies generalize the Shannon entropy and constitute a family of one-parameter information measures. In the quantum setting, the corresponding concepts are the von Neumann \cite{von1932mathematische} and quantum R\'enyi entropies \cite{Petz1986a}, {which have applications in many fields such as quantum chemistry \cite{aspuru2005simulated}, condensed matter physics \cite{laflorencie2016quantum}, and high energy physics \cite{Peschanski2019}.}
{In particular,} computing quantum entropies plays a key role in quantum information and quantum computing. For instance, quantum entropies can provide the asymptotic lower bound for compressing quantum data~\cite{schumacher1995quantum} and be applied to study quantum Gibbs state preparation \cite{Chowdhury2020,Yuan2018a,Wu2019b,Wang2020} and Hamiltonian learning \cite{Wiebe2013,Anshu2021,Bairey2018,wang2021hybrid}. 

For any quantum state $\rho\in\mathbb{C}^{2^n\times2^n}$, the von Neumann entropy is defined by $S(\rho)\coloneqq -\tr(\rho\ln(\rho))$, and the quantum $\alpha$-R\'enyi entropy is defined by $R_{\alpha}(\rho)\coloneqq\frac{1}{1-\alpha}\log\tr(\rho^\alpha)$ with parameter $\alpha\in(0,1)\cup(1,+\infty)$. Taking the limit $\alpha\to1$, $R_{\alpha}(\rho)$ converges to $S(\rho)$ up to a proportional factor. Additionally, if $\rho$ is diagonal in the computational basis, $S(\rho)$ and $R_{\alpha}(\rho)$ degenerate to the classical counterparts. 
Various methods \cite{Montanaro2016} have been proposed to estimate quantum entropies in past decades, while a large number of quantum resources are demanded as well.
The most straightforward method to estimate quantum entropy is using tomography \cite{Acharya2020}, which figures out the description of the density matrix. In that case, the consumption increases exponentially with the size of the state. On top of that, the current optimal classical algorithm for quantum entropy estimation has a cost that is linear to the number of non-zero elements of the density matrix~\cite{Kontopoulou2018}.

Regarding quantum computing methods, many proposals based on different models have been proposed~\cite{Hastings2010,Acharya2020,Li2019a,Gilyen2019,Subramanian2019,Chowdhury2020,Luongo2020}. Specifically speaking, \cite{Acharya2020} studies the cost of estimating the von Neumann and R\'enyi entropies in a model where one can get independent copies of the state. By allowing arbitrary measurements and classical post-processing, it shows that the cost of entropy estimation scales exponentially in the state size.
Later, \cite{Gilyen2019} and \cite{Subramanian2019} study the von Neumann entropy and quantum R\'enyi entropy estimation in a quantum query model, respectively. The query model here is a quantum circuit that can prepare the purification of the input state, i.e., $U_{\rho}\ket{0}_A\ket{0}_B=\ket{\psi_\rho}_{AB}$ and $\tr_{A}(\op{\psi_\rho}{\psi_\rho}_{AB})=\rho$. 
In these works, the times of using $U_{\rho}$ for estimating $S(\rho)$ could be linear in the dimension ($O(d)$ \cite{Gilyen2019}, where $d$ is the dimension of system), while the results for $R_\alpha(\rho)$ is comparable to the tomography ($O(d^2)$) \cite{Subramanian2019}. 
{Recently, the work \cite{Gur2021} has brought the number of using $U_\rho$ for estimating $S(\rho)$ to be sub-linear in the system's dimension.} Another work \cite{Yirka2020} has considered access to the purification of a state and used short-depth circuits, which generalize the swap test, to estimate $\tr(\rho^k)$. 

Although aforementioned quantum algorithms have promised quantum speedups, the quantum query model for the input state, the most crucial component of the algorithms, is still not known how to construct efficiently. And hence, the timescale for these algorithms to be effective in practice remains an open question. On the other hand, the fast development of quantum computing devices has brought us into the noisy intermediate-scale quantum (NISQ) era \cite{Preskill2018}. An important research direction is to exploit NISQ devices to solve challenging tasks for classical computers. To better exploit NISQ devices in the quantum entropy estimation task, it is highly desirable to devise quantum algorithms without using the quantum query model.

In this work, we propose quantum algorithms of concrete implementation to estimate the von Neumann and quantum R\'enyi entropies of an unknown quantum state, using independent copies of the input state.
To develop our algorithms, we firstly use the Fourier series approximation to decompose the entropy. Then, we devise quantum circuits to estimate individual term in the series. When designing quantum circuits, we synthesize several quantum gadgets, such as the iterative quantum phase estimation~\cite{kitaev1995quantum}, the exponentiation of the quantum state in \cite{Lloyd2014}, the linear combination of unitaries \cite{Berry2015}, and qubit reset~\cite{egger2018pulsed}.  {As a result, the circuits are composed of primitive single/two-qubit gates, and two copies of the state are maintained during computation.}
In the end, we could obtain the estimated entropy by classical post-processing. Particularly, we utilize the sampling method to reduce the computational resources and speed up computation. 
 

Compared with algorithms of~\cite{Gilyen2019,Subramanian2019,Chowdhury2020,Luongo2020,Gur2021}, our circuits no longer depend on the quantum query model but use copies of the input state. Generally, constructing such a model is almost as hard as state tomography. In contrast, our algorithms need less information (about the minimal non-zero eigenvalue) and thus are more implementable in practice. Moreover, implementing algorithms of~\cite{Gilyen2019,Subramanian2019,Chowdhury2020,Luongo2020,Gur2021} needs to devise corresponding circuits for different quantum states. In comparison, our methods can use fixed circuits to estimate all selected states. As a result, when implemented in experiments, our methods are more resource-efficient in compiling logical qubits to physical qubits. Thirdly, when we deal with special cases where the minimal non-zero eigenvalue is polynomially small, i.e., $\Lambda=\Omega(1/poly(n))$, our algorithm consumes polynomial resources while algorithms of \cite{Gilyen2019,Subramanian2019,Luongo2020,Gur2021} require exponential resources. Fourthly, the approach in \cite{Yirka2020} can estimate the quantum R\'enyi entropy $R_{\alpha}(\alpha)$ when the parameter $\alpha$ is integer. In comparison, our approach is more general, i.e., our approach can apply to the case where $\alpha$ is an integer or non-integer.

\section{Overview of our results}\label{sec:problem}
In this paper, we assume that copies of the input state $\rho$ can be accessed and have no constraint on the number of copies. Then, we formally state the task of estimating quantum entropies below.

\begin{definition}[Entropy estimation]
Given free access to the copies of a quantum state $\rho\in\mathbb{C}^{2^n\times 2^n}$, the aim is to estimate the von Neumann entropy $S(\rho)=-\tr(\rho\ln\rho)$ and $\alpha$-R\'enyi entropy $R_{\alpha}(\rho)=\frac{1}{1-\alpha}\log\tr(\rho^\alpha)$. To be more specific, find $S(\rho)_{est}$ and $R_{\alpha}(\rho)_{est}$ such that, for any constant $\alpha\in(0,1)\cup(1,+\infty)$,
\begin{align}
    &\Pr\left[\left|S(\rho)_{ est}-S(\rho)\right|\leq\epsilon\right]\geq1-\delta,\\
    &\Pr[|R_{\alpha}(\rho)_{est}-R_{\alpha}(\rho)|\leq\epsilon]\geq 1-\delta,
\end{align}
where $\epsilon\in(0,1)$ and $\delta\in(0,1)$ denote the estimation accuracy and the failure probability, respectively.
\end{definition}

To realize the defined tasks on quantum computers, the main idea is to find a Fourier series approximation of the entropy and evaluate the Fourier series by constructing explicit quantum circuits. In particular, we establish the following:
In Sec. \ref{sec:fourer_series}, we propose the Fourier series as approximations of the von Neumann and R\'enyi entropies, which means we decompose the von Neumann and quantum R\'enyi entropy into the combination of many terms that are easy to estimate. In Sec. \ref{sec:quantum_circuit}, we provide explicit quantum circuits to evaluate the Fourier series of the entropy approximations.
In Sec. \ref{sec:quantum_algorithm}, combining the Fourier series approximation and explicit circuit schemes, we propose quantum algorithms for estimating von Neumann and quantum R\'enyi entropies.
To demonstrate the effectiveness of our algorithms, we conduct numerical experiments in Sec.~\ref{sec:experiment}. Meanwhile, we also study the robustness of our algorithms to depolarizing and amplitude damping noise channels. 
Last, we compare our algorithms with the existing approaches and discuss several applications in Sec.~\ref{sec:Discussion}. Finally, the paper is concluded in Sec. \ref{sec:conclusion}. 

\section{Quantum entropy approximations}\label{sec:fourer_series}
Fourier series approximations of quantum entropies have been previously considered in \cite{Gilyen2019,Chowdhury2020,Subramanian2019}. Our results are inspired by the work \cite{Chowdhury2020}, which employs a method in \cite{van2017quantum} to convert a Taylor series approximation of $S(\rho)$ to Fourier series. In this section, we further employ this method to give approximations of $S(\rho)$ and $R_{\alpha}(\rho)$, expressed as a linear combination of terms of the form $\tr(\rho\cos(\rho t))$.

\subsection{Approximation of von Neumann entropy}
To provide the series approximation, we follow the method given in Lemma 37 of \cite{van2017quantum} to construct the Fourier series from a truncated Taylor series. Here, we use the truncated Taylor series of $S(\rho)$ shown below.
\begin{align}
    S(\rho)\approx\sum_{k=1}^{K}\frac{1}{k}\tr\left(\rho(I-\rho)^k\right),
\end{align}
where integer $K$ is the truncation order determining the accuracy. {The larger $K$ gives a more accuracy entropy $S(\rho)$.} The details of derivation are deferred to Appendix~\ref{sec:series}. Using this Taylor series, we can find a Fourier series approximation $S(\rho)_{est}$, which is presented in Lemma \ref{le:fourier}. 

\begin{lemma}\label{le:fourier}
For arbitrary quantum state $\rho\in\mathbb{C}^{2^n\times 2^n}$, let $\Lambda$ be the lower bound on all non-zero eigenvalues of $\rho$. There exists a Fourier series $S(\rho)_{est}$ such that $\left|S(\rho)-S(\rho)_{est}\right| \leq\epsilon$ for any $\epsilon\in(0,1)$, where
\begin{align}
    S(\rho)_{est} = \sum_{l=0}^{\lfloor L\rfloor}\sum_{s=D_l}^{U_l}\sum_{k=1}^{K}\frac{b_{l}^{(k)}\binom{l}{s}}{k2^{l}}\tr(\rho\cos(\rho\cdot t(s,l))).\label{eq:entropy estimation}
\end{align}
Particularly, coefficient $U_{l}=\min\{l,\lceil \frac{l}{2}\rceil+M_l\}$ and $D_l=\max\{0,\lfloor \frac{l}{2}\rfloor-M_l\}$, and $\binom{l}{s}$ denotes the binomial coefficient. Meanwhile, coefficients $t(s,l)=(2s-l)\pi/2$, and coefficients $K,L,M_l$ are given by 
\begin{align}
K\in\Theta\left(\frac{\log(\epsilon\Lambda)}{\log(1-\Lambda)}\right), \quad L=\ln\left(\frac{4\sum_{k=1}^{K}1/k}{\epsilon}\right)\frac{1}{\Lambda^2}, \quad M_l=\left\lceil\sqrt{\ln\left(\frac{4\sum_{k=1}^{K}1/k}{\epsilon}\right)\frac{l}{2}}\right\rceil.
\end{align}
For any $k=1,\ldots,K,l=0,\ldots,\lfloor L\rfloor$, the coefficients $b_{l}^{(k)}$ are positive and defined inductively. Explicitly, 
\begin{align}
      b_{l}^{(1)}=0,\quad\text{if $l$ is even},\quad b_{l}^{(1)}=\frac{2\binom{l-1}{(l-1)/2}}{\pi2^{l-1}l},\quad\text{if $l$ is odd},\quad b_{l}^{(k+1)}=\sum_{l'=0}^{l}b_{l'}^{(k)}b_{l-l'}^{(1)}, \quad \forall k\geq1. \label{eq:blk}
\end{align}
Moreover, overall weights of $S(\rho)_{est}$ is bounded as follows,
\begin{align}
  \sum_{l=0}^{\lfloor L\rfloor}\sum_{s=D_{l}}^{U_l}\sum_{k=1}^{K}\frac{b_{l}^{(k)}\binom{l}{s}}{2^lk}\in O\left(\log(K)\right).
\end{align}
\end{lemma}
\begin{proof}[Sketch of the proof]
Briefly speaking, we transform the truncated Taylor series $\sum_{k=1}^{K}\frac{1}{k}\tr\left(\rho(I-\rho)^k\right)$ into a weighted sum of cosines by first substituting $1-\rho$ with $\arcsin(\cos(\rho\pi/2))/\pi/2$ and then expanding $\arcsin$ to its Taylor series. Subsequently, we truncate the Taylor series of $\arcsin$ to derive a high-precision approximation. Last, we use the relation $\cos(\rho\pi/2)=\frac{e^{i\rho\pi/2}+e^{-i\rho\pi/2}}{2}$ to cancel the cosines, resulting the desired Fourier series $S(\rho)_{est}$. We provide the detailed analysis in Appendix~\ref{sec:fourier}.
\end{proof}

\subsection{Approximation of $\alpha$-R\'enyi entropy}
As for the quantum $\alpha$-R\'enyi entropy, the estimation task can be simplified. 
Notice the expression $R_{\alpha}(\rho)$, we only need to focus on the quantity $\tr(\rho^\alpha)$ and derive the desired estimate via calculation after obtaining the estimate of $\tr(\rho^\alpha)$. So, we construct the Fourier series approximation to $\tr(\rho^\alpha)$. Moreover, we display how the estimation error of $\tr(\rho^\alpha)$ propagates to the $\alpha$-R\'enyi entropy. 

To begin with, we write $\tr(\rho^\alpha)=\tr(\rho\cdot\rho^{\beta})$, where $\beta=\alpha-1$. Recall the Taylor series of the power function $x^{\beta}$ over the interval $x\in(0,2)$. 
\begin{align}
    x^\beta=\sum_{k=0}^{\infty}\binom{\beta}{k}(x-1)^k,
\end{align}
where $\binom{\beta}{k}=\prod_{j=1}^{k}\frac{\beta-j+1}{j}=\frac{\beta(\beta-1)\ldots(\beta-k+1)}{k!}$ is the generalized binomial coefficient.

Clearly, truncating the infinity series of $x^\beta$ would lead to the desired Taylor series. While, for different $\beta$, the truncation order $K$ will be various. To get such an order, we need more information about the generalized binomial coefficients. Thus we show the bounds on the generalized binomial coefficient below. Please note that the proofs of following results are provided in Appendix~\ref{app:proposition1} \& \ref{app:proposition2} \& \ref{app:renyi_fourier}.
\begin{proposition}\label{le:bound_binom}
For any constant $\beta\in(-1,0)\cup(0,+\infty)$, there exists a bound on the generalized binomial coefficient $\binom{\beta}{k}$. 
\begin{enumerate}
    \item For $\beta\in(-1,0)$ and any integer $k\geq 1$, $|\binom{\beta}{k}|\leq |\beta|$. 
    \item For $\beta\in(0,1]$ and any $k\geq 2$, $\left|\binom{\beta}{k}\right|\leq \left[1+\frac{\beta\ln\frac{(k+1)}{k^2}+\beta-1}{k}\right]^k.$
Particularly, $|\binom{\beta}{k}|\leq \frac{1}{e}$, if $k\in\Omega(1)$. 
\item For $\beta\in(1,+\infty)$ and $k\geq \beta+1$, $\left|\binom{\beta}{k}\right|\leq \left[1+\frac{\beta\ln\frac{(\beta+1)^2}{k}+2}{k}\right]^k.$
Particularly, $|\binom{\beta}{k}|\leq 1$, if $k\in \Omega((\beta+1)^2)$.
\end{enumerate}

Moreover, for an integer $K$, the sum $\sum_{k=1}^{K}\left|\binom{\beta}{k}\right|$ is bounded. 
\begin{align}
    \sum_{k=1}^{K}\left|\binom{\beta}{k}\right|\leq\left\{
    \begin{array}{ll}
     O\left( K+e^2(\beta+1)^{2\beta}\cdot\left[\ln(\lceil e^{\frac{2}{\beta}}(\beta+1)^2\rceil+1)+1\right]-e^{\frac{2}{\beta}}(\beta+1)^2 \right),  & \text{if $\beta\in(1,+\infty)$}, \\
        O(K), & \text{if $\beta\in(0,1]$},\\
     O\left(   |\beta| K \right), & \text{if $\beta\in(-1,0)$}.
    \end{array}\right.
\end{align}
\end{proposition}

With the bounds on $\binom{\beta}{k}$ in Proposition~\ref{le:bound_binom}, we are able to develop the Taylor series approximation to $\tr(\rho^\alpha)$.
\begin{proposition}\label{le:taylor_series}
For any constant $\alpha\in(0,1)\cup(1,+\infty)$ and $\xi\in(0,1)$, there exists an integer $K$ such that, for any quantum state $\rho$ with eigenvalue lower bound $\Lambda$,
\begin{align}
    \left|\tr(\rho^\alpha)-1-\sum_{k=1}^{K}\binom{\beta}{k}\tr\left(\rho(\rho-I)^k\right)\right|\leq\xi.\label{eq:taylor_approximation}
\end{align}
Particular, the choice of integer $K$ is shown below.
\begin{align}
K\in\left\{
    \begin{array}{ll}
        \max\{\Omega(\alpha^2), \Omega(\log(\Lambda\xi)/\log(1-\Lambda))\} & \text{if $\alpha\in(2,+\infty)$}, \\[0.5em]
        \Omega(\log(\Lambda\xi)/\log(1-\Lambda)) &  \text{if $\alpha\in(0,1)\cup(1,2]$}.
    \end{array}\right.
\end{align}
\end{proposition}

Now we could give a Fourier series approximation of $R_{\alpha}(\rho)$.
\begin{lemma}\label{le:renyi_fourier}
Consider a quantum state $\rho\in\mathbb{C}^{2^n\times 2^n}$. Let $\Lambda\in(0,1)$ be a lower bound on all non-zero eigenvalues. For any $\alpha\in(0,1)\cup(1,+\infty)$, there exists an estimate $R_{\alpha}(\rho)_{est}$ of $R_{\alpha}(\rho)$ up to precision $\epsilon$. To be specific,
\begin{align}
    R_{\alpha}(\rho)_{est}&=\frac{1}{1-\alpha}\log F_{\alpha}(\rho),
\end{align}
and $F_{\alpha}(\rho)$ satisfies $|F_{\alpha}(\rho)-\tr(\rho^\alpha)|\leq \xi$, where
\begin{align}
  F_{\alpha}(\rho)= 1+ \sum_{l=0}^{\lfloor L\rfloor}\sum_{s=D_l}^{U_l}\left(\sum_{k=1}^{K}(-1)^kb_{l}^{(k)}\binom{\alpha-1}{k}\right)2^{-l}\binom{l}{s}\tr(\rho\cdot \cos(\rho t(s,l))). \label{eq:renyi_approximation}
\end{align}
In particular, the relation between $\epsilon$ and $\xi$ is given in Eq.~\eqref{eq:renyi_precision}, and definition of all $b_{l}^{(k)}$ are given in Eq.~\eqref{eq:blk}.
And the parameters of $F_{\alpha}(\rho)$ are given as follows. Coefficient $t(s,l)=\frac{(2s-l)\pi}{2}$, and $U_{l}=\min\{l,\lceil \frac{l}{2}\rceil+M_l\}$ and $D_l=\max\{0,\lfloor \frac{l}{2}\rfloor-M_l\}$. Moreover,
\begin{align}
    K=\Theta\left(\frac{\log(\Lambda\xi)}{\log(1-\Lambda)}+\alpha^2\right),\quad L=\ln\left(\frac{4\sum_{k=1}^{K}|\binom{\alpha-1}{k}|}{\xi}\right)\frac{1}{\Lambda^2},\quad M_l=\left\lceil\sqrt{\ln\left(\frac{4\sum_{k=1}^{K}|\binom{\alpha-1}{k}|}{\xi}\right)\frac{l}{2}}\right\rceil.
\end{align}
And the overall weights of $F_{\alpha}(\rho)$ are bounded by $\sum_{k=1}^{K}|\binom{\alpha-1}{k}|$, and bounds on $\sum_{k=1}^{K}|\binom{\alpha-1}{k}|$ are given in Table~\ref{tab:my_label_weights}.

\begin{table}[h]
    \centering
    \begin{tabular}{|c|c|}
    \hline
    $\alpha$ & Bound on $\sum_{k=1}^{K}|\binom{\alpha-1}{k}|$ \\[0.2em]
    \hline
        $(0,1)$  & $O(|\alpha-1|K)$\\[0.2em]
        \hline
         $(1,2]$ & $O(K)$ \\[0.2em]
         \hline
        $(2,+\infty)$ & $O\left( K+e^2(\alpha)^{2\alpha-2}\cdot\left[\ln(\lceil e^{\frac{2}{\alpha-1}}(\alpha)^2\rceil+1)+1\right]-e^{\frac{2}{\alpha-1}}(\alpha)^2 \right)$\\[0.2em]
        \hline
    \end{tabular}
    \caption{Upper bound on the overall weights.}
    \label{tab:my_label_weights}
\end{table}
\end{lemma}

Here, we discuss the relation between $\xi$ and $\epsilon$. Let $\widehat{\tr(\rho^\alpha)}$ be an estimate of $\tr(\rho^\alpha)$ up to error $\xi$, i.e., $|\widehat{\tr(\rho^\alpha})-\tr(\rho^\alpha)|\leq\xi$. Then the difference between the corresponding logarithms is given below.
\begin{align}
    \left|\frac{1}{1-\alpha}\log\widehat{\tr(\rho^\alpha)}-\frac{1}{1-\alpha}\log\tr(\rho^\alpha)\right|=\frac{1}{|1-\alpha|}\left|\log\frac{\widehat{\tr(\rho^\alpha)}}{\tr(\rho^\alpha)}\right|=\frac{1}{|1-\alpha|}\left|\log\left(1+\frac{\widehat{\tr(\rho^\alpha})-\tr(\rho^\alpha)}{\tr(\rho^\alpha)}\right)\right|.
\end{align}
Notice that $|\log(1+x)|\leq 2|x|$ for any $x\in[-0.5, 1]$. Then we can assume that $\frac{\widehat{\tr(\rho^\alpha})-\tr(\rho^\alpha)}{\tr(\rho^\alpha)}$ falls in the interval $[-0.5,1]$ and hence we have
\begin{align}
     \left|\frac{1}{1-\alpha}\log\tr(\rho^\alpha)-\frac{1}{1-\alpha}\log\widehat{\tr(\rho^\alpha)}\right| \leq \frac{2}{|1-\alpha|}\left|\frac{\widehat{\tr(\rho^\alpha})-\tr(\rho^\alpha)}{\tr(\rho^\alpha)}\right|\leq \frac{2\xi}{|1-\alpha||\tr(\rho^\alpha)|}.
\end{align}
Moreover, since $\tr(\rho^\alpha)\geq [\tr(\rho^2)]^{\alpha-1}$ for all $\alpha\in(0,1)\cup(2,+\infty)$, and $\tr(\rho^{\alpha})\geq\tr(\rho^2)$ for all $\alpha\in(1,2]$, we can determine $\xi$ upon receiving $\epsilon$. Explicitly, $\xi$ is given by
\begin{align}
    \xi=\left\{
    \begin{array}{ll}
      \frac{|1-\alpha|[\tr(\rho^2)]^{\alpha-1}}{2}\epsilon,  & \forall\alpha\in(0,1)\cup(2,+\infty), \\
        \frac{|1-\alpha|\tr(\rho^2)}{2}\epsilon, & \forall\alpha\in(1,2].
    \end{array}\right. \label{eq:renyi_precision}
\end{align}

\begin{remark}
The Swap test~\cite{ekert2002direct} can evaluate the term $\tr(\rho^2)$ efficiently. As a result, the 2-R\'enyi entropy can be obtained via Swap test as well. If let $r_\rho$ be the rank of state $\rho$, we can substitute $\tr(\rho^2)$ with $1/r_{\rho}$ due to the fact that $\tr(\rho^2)\geq 1/r_{\rho}$. In these cases, the estimation accuracy $\xi$ could be polynomially small if $\rho$ is low-rank or $\tr(\rho^2)$ is polynomially small.
\end{remark}

Now, we have provided estimates for the von Neumann and $\alpha$-R\'enyi entropy in Lemma~\ref{le:fourier} \& \ref{le:renyi_fourier}. Especially, these estimates can be obtained by evaluating the Fourier series in Eq.~\eqref{eq:entropy estimation} \& \eqref{eq:renyi_approximation} on quantum computers. To achieve this purpose, we proceed to devise quantum circuits to estimate quantity $\tr(\rho\cos(\rho\cdot t(s,l)))$.

\section{Quantum circuits}\label{sec:quantum_circuit}
In this section, we first show a scheme to estimate the term $\tr(\rho\cos(\rho t))$. We also demonstrate the validity and estimate the cost of primitive single/two-qubit gates. Then we discuss compressing the circuit width. In the end, we discuss a crucial subroutine that can simulate the exponentiation of the Swap operator. 

\subsection{Circuit scheme}
For simplicity, we first consider estimating $\tr(\rho\cos(\rho t))$ with small $t$. Based on the circuit of iterative quantum phase estimation \cite{kitaev1995quantum}, the circuit for this purpose is depicted in Figure~\ref{figure:circuit_exponentiation}. Please note that we denote the top first qubit of the circuit as the measure register. The first state $\rho$ is prepared in the main register, and other copies are prepared on the ancillary registers.
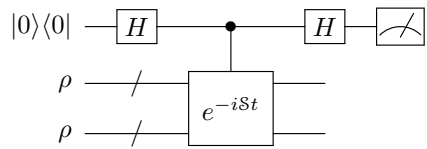
\begin{figure}[htb]
\[ 
\Qcircuit @C=1.2em @R=1em {
\lstick{\op{0}{0}} & \gate{H} & \ctrl{1} & \gate{H} & \meter\\
\lstick{\rho} &/\qw & \multigate{1}{e^{-i\mathcal{S} t}} & \qw &\\
\lstick{\rho} &/\qw & \ghost{e^{-i\mathcal{S} t}} & \qw &
}
\]
\caption{\footnotesize For a short time $t$, we first prepare a ground state $\op{0}{0}$ in the measure register, and prepare states $\rho$ in the main register the ancillary register, respectively. Subsequently, perform the controlled unitary operator $e^{-i\mathcal{S}t}$ on state $\rho\otimes\rho$. At the end of the circuit, we measure along the eigenbasis of Pauli $Z$, which would immediately lead to an estimate for $\tr(\rho\cos(\rho t))$ up to precision $O(t^2)$.}
\label{figure:circuit_exponentiation}
\end{figure}

In Figure~\ref{figure:circuit_exponentiation}, qubit $\op{0}{0}$ and two copies of state $\rho$ are input into the circuit. Then two Hadamarad gates are applied to the measure register, sandwiching a controlled operation operation c-$e^{-i\mathcal{S} t}$, i.e., the exponentiation of the Swap operator, where $\mathcal{S}$ denotes the Swap operator. In the end of the circuit, the measurement occurs on the measure register along the eigenbasis of the Pauli matrix $Z$. Particularly, the measurement outcomes lead to the value of $\tr(\rho\cos(\rho t))$. More precisely, $\tr(\rho\cos(\rho t))\approx\Pr[0]-\Pr[1]$, where $\Pr[0/1]$ denotes the probability of observing outcome $0/1$, respectively. And the estimation accuracy is shown in the result below.
\begin{proposition}\label{th:main_circuit}
Let $V$ denote the unitary corresponding to the circuit in Figure~\ref{figure:circuit_exponentiation}. For any quantum state $\rho$ and a $small$ parameter $t\in(-1,1)$, the measurement outcome is close to $\tr(\rho\cos(\rho t))$. Explicitly,
\begin{align}
    \left|\tr(\rho\cos(\rho t))-\tr\left(\left(V(\op{0}{0}\otimes\rho^{\otimes2})V^{\dagger}\right)Z_0\right)\right|\leq 2t^2.
\end{align}
Here, $Z_0$ equals to $Z \otimes I \otimes I$, which indicates measuring the measure register along Pauli $Z$'s eigenbasis.
\end{proposition}

\begin{proof}
To demonstrate the validity, we focus on the state of the measure register and the main register. Notice that, after the controlled operation, the initial state evolves into the following form.
\begin{align}
    \frac{1}{2}\left[\op{0}{0}\otimes\rho^{\otimes 2}+\op{1}{1}\otimes e^{-i\mathcal{S}t}\rho^{\otimes 2}e^{i\mathcal{S}t}\right]+\frac{1}{2}\left[\op{0}{1}\otimes\rho^{\otimes 2}e^{i\mathcal{S}t}+\op{1}{0}\otimes e^{-i\mathcal{S}t}\rho^{\otimes 2}\right].\label{eq:state_evolved}
\end{align}

Next, trace out the appended copy of $\rho$. For instance, a part of Eq.~\eqref{eq:state_evolved} is traced as follows.
\begin{align}
    |1\rangle\langle 0|\otimes \tr_{anc}\left(e^{-i \mathcal{S}t}\rho^{\otimes 2}\right)=|1\rangle\langle 0|\otimes \tr_{anc}\left(\left(\cos(t) I-i\sin(t)\mathcal{S}\right) \rho^{\otimes 2}\right)=\op{1}{0}\otimes(\cos(t)I-i\sin(t)\rho)\rho.\label{eq:01}
\end{align}
Notation $\tr_{anc}$ means tracing out the appended copy of $\rho$. Here, we have used facts that $e^{-i\mathcal{S}t}=\cos(t)I-i\sin(t)\mathcal{S}$ and $\tr_{anc}(\mathcal{S}\rho\otimes \sigma)=\rho\sigma$. A similar result could be derived for $\op{0}{1}\otimes\rho^{\otimes2} e^{i\mathcal{S}t}$.

We bound the difference between $(\cos(t)I-i\sin(t)\rho)\rho$ and $e^{-i\rho t}\rho$.
\begin{align}
    {\rm difference}&=\left\|\left[e^{-i\rho t}-(\cos(t)I-i\sin(t)\rho)\right]\rho\right\|_{tr}\label{eq:difference}\\
    &=\left\|(1-\cos(t))\rho-i(t-\sin(t))\rho^2+\sum_{k\geq 2}\frac{(-i\rho t)^k}{k!}\rho\right\|_{tr}\nonumber\\
    &\leq 2\sin^2(t/2)+|t-\sin(t)|+\frac{\sqrt{2}t^2}{2}\nonumber\\
    &\leq 2t^2,\nonumber
\end{align}
where we have used inequalities $|x-\sin(x)|\leq x^2/2$ and $|\sin(x)|\leq x$, and the inequality $\|\sum_{k\geq2}{(-i\rho t)^k}/{k!}\|_{tr}\leq {\sqrt{2}t^2}/{2}$ (The proof can be found in Appendix~\ref{sec:exponent_inequality}). Again, a similar result can be found for $\rho^{\otimes2} e^{i\mathcal{S}t}$ and $\rho e^{i\rho t}$. 

Note that the measurement outcome is 
\begin{align}
    \Pr[0]-\Pr[1]&=\frac{\tr\left(e^{-i\mathcal{S}t}\rho^{\otimes2}\right)+\tr\left(\rho^{\otimes2}e^{i\mathcal{S}t}\right)}{2}=\frac{\tr\left(\tr_{anc}(e^{-i\mathcal{S}t}\rho^{\otimes2})\right)+\tr\left(\tr_{anc}(\rho^{\otimes2}e^{i\mathcal{S}t})\right)}{2}.
\end{align}
Recall we have shown that $\left\|\tr_{anc}(e^{-i\mathcal{S}t}\rho^{\otimes2})-e^{-i\rho t}\rho\right\|_{tr}\leq 2t^2$ and $\left\|\tr_{anc}(\rho^{\otimes2}e^{i\mathcal{S}t})-\rho e^{i\rho t}\right\|_{tr}\leq 2t^2$.
Besides, $\tr(\rho\cos(\rho t))=[\tr(e^{-i\rho t}\rho)  +\tr(\rho e^{i\rho t})]/2$. Immediately, we derive the result
\begin{align}
    \left|\Pr[0]-\Pr[1]-\tr(\rho\cos(\rho t))\right|\leq 2t^2.
\end{align}
\end{proof}

Now we have shown that the circuit in Figure~\ref{figure:circuit_exponentiation} can be used to estimate the term $\tr(\rho\cos(\rho t))$, especially when $t$ is small.
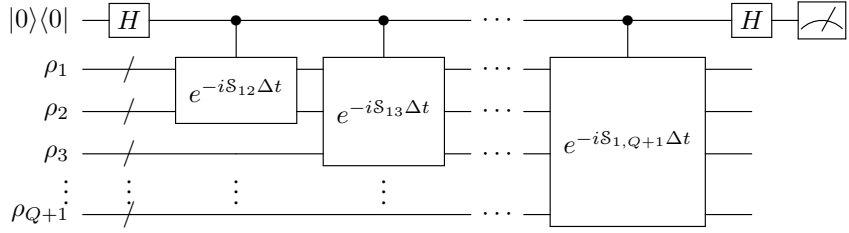
\begin{figure}[htb]
\[ 
\Qcircuit @C=1em @R=0.7em {
\lstick{\op{0}{0}} & \gate{H} & \ctrl{1} & \ctrl{1} &\qw & \cdots & & \ctrl{1} & \gate{H} & \meter\\
\lstick{\rho_{1}} &/\qw & \multigate{1}{e^{-i\mathcal{S}_{12}\Delta t}} & \multigate{2}{e^{-i\mathcal{S}_{13}\Delta t}} &\qw& \cdots & & \multigate{4}{e^{-i\mathcal{S}_{1,Q+1}\Delta t}} & \qw &\\
\lstick{\rho_{2}} &/\qw & \ghost{e^{-i\mathcal{S}_{12}\Delta t}}  & \ghost{e^{-i\mathcal{S}_{13}\Delta t}}&\qw & \cdots & & \ghost{e^{-i\mathcal{S}_{1,Q+1}\Delta t}} & \qw &\\
\lstick{\rho_{3}} &/\qw & \qw  & \ghost{e^{-i\mathcal{S}_{13}\Delta t}}&\qw & \cdots & & \ghost{e^{-i\mathcal{S}_{1,Q+1}\Delta t}}& \qw \\
\lstick{\vdots} &\vdots & \vdots & \vdots& &  & &   & &  & &  \\
\lstick{\rho_{Q+1}} &/\qw & \qw  & \qw&\qw &\cdots &  & \ghost{e^{-i\mathcal{S}_{1,Q+1}\Delta t}}& \qw
}
\]
\caption{\footnotesize 
For general time $t$, the circuit could be inductively constructed. The operator $e^{-i\mathcal{S}\Delta t}$ is sequentially applied on the main register and different ancillary registers, conditional on the measure register. Here, we append $Q$ ancillary states and use $Q$ times of $e^{-i\mathcal{S}\Delta t}$. For clear, we label states on different register by $1,2,3,Q+1$, and the script of the swap operator indicates the registers that swap operator acts on. 
}
\label{figure:circuit_instance}
\end{figure}
Regarding a large $t$, we use a circuit similar to that in Figure~\ref{figure:circuit_exponentiation}. Particularly, we divide the parameter $t$ into several small pieces {$\Delta t$} and run c-$e^{-i\mathcal{S}\Delta t}$ sequentially. The corresponding circuit is depicted in Figure~\ref{figure:circuit_instance}, where we run $Q$ controlled operations, and the parameter $\Delta t$ is small.

Based on the result in Proposition~\ref{th:main_circuit}, we can readily derive the result for estimating $\tr(\rho\cos(\rho t))$ with a large $t\in\mathbb{R}$.
\begin{proposition}\label{th:main_cost}
For any quantum state $\rho\in\mathbb{C}^{2^n\times2^n}$ and $t\in\mathbb{R}$, there is a quantum circuit that can estimate the quantity $\tr(\rho\cos(\rho t))$ with precision $\epsilon$. The number of needed copies of state is $O(t^2/\epsilon)$.
\end{proposition}
\begin{proof}
Suppose we divide the parameter $t$ into $Q$ pieces and the parameter in the Swap operator becomes $\Delta t=t/Q$. In the circuit, we need $Q$ ancillary register and the state in each ancillary register is $\rho$. For each interval of length $\Delta t$, the resulting error in the current state of the measure register and main register is at most $2\Delta t^2$ by Proposition \ref{th:main_circuit}. Easily, we can deduce the accumulating error, i.e., $Q \times 2\Delta t^2\leq 2t^2/Q$. To suppress the overall accumulating errors, we set $Q=\lceil\frac{2t^2}{\epsilon}\rceil$. Consequently, the final error is at most $\epsilon$. 

Moreover, the number of used copies is $(Q+1)=O(t^2/\epsilon)$.
\end{proof}

Proposition~\ref{th:main_cost} has shown that we can use the circuit in Figure~\ref{figure:circuit_instance} to estimate the Fourier series. While there are some remaining issues. One is that the large circuit width may be a huge burden in practice. Another is to simulate the exponentiation of the Swap operator. The solutions to overcome these issues are discussed in the following sections.

\subsection{Circuit width circumvent}\label{sec:qubit_reset}
In Figure~\ref{figure:circuit_instance}, there are $Q+1$ copies of $\rho$ prepared at the beginning, while the interaction only occurs between two copies at a time. For instance, copy $\rho_{1}$ only interacts with the ancillary register being state $\rho_2$. Once the controlled operation c-$e^{-i\mathcal{S}_{12}\Delta t}$ is completed, the occupied ancillary register is relieved. At that time, copy $\rho_3$ is employed for the next interaction. Hence, interactions occur alternatively between the measure register and main register and different ancillary registers.

Notice that the relieved ancillary register will no longer affect the state of the rest registers. Thus, we can measure the relieved register and reuse it for preparing a new state by the qubit reset technique. Using qubit reset means that we can measure subsets of the qubits and reinitialize them~\cite{egger2018pulsed}. In the past decade, many experimental methods have been developed to actively reset the qubits on superconducting qubits~\cite{reed2010fast,riste2012initialization,geerlings2013demonstrating,govia2015unitary,egger2018pulsed,magnard2018fast}. Recently, the qubit reset also applies to design quantum algorithms~\cite{Huggins2019a,Liu2019a,foss2021holographic,Yirka2020,Rattew2020} with reduced circuit width. Particularly, \cite{Yirka2020} uses qubit reset to devise quantum circuits for estimating $\tr(\rho^{ k})$ with $k\in\mathbb{N}$, which can contribute to the problem of quantum R\'enyi entropy \cite{renyi1961measures} estimation. 

Here, we use the qubit reset to compression the width of the circuit in Figure~\ref{figure:circuit_instance}. In this case, we only need one ancillary register. Specifically, we prepare the state $\rho_1\otimes\rho_2$ on the main register and ancillary register in the beginning. Then the interaction occurs between them. Once the interaction ends, the ancillary register is measured. Subsequently, the ancillary register is readily reset to state $\rho_{3}$. And then, the interaction occurs again. The same procedure of measuring and the reset repeats $Q$ times in all. The corresponding circuit using qubit reset can be found in Figure~\ref{figure:circuit_instance_reset}.

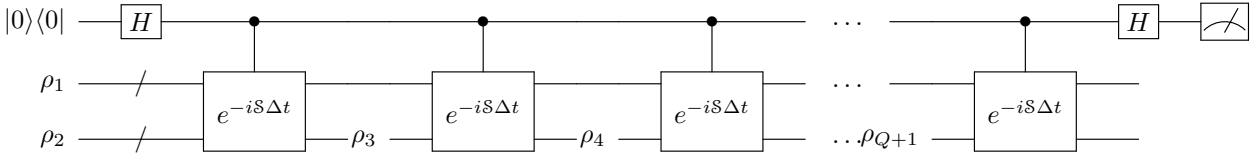
\begin{figure*}[hbt]
\[ 
\Qcircuit @C=1.6em @R=1.2em {
\lstick{\op{0}{0}} & \gate{H} & \ctrl{1} & \qw & \qw & \ctrl{1} & \qw & \qw & \ctrl{1} & \qw & \ldots &  & \qw & \ctrl{1} &\gate{H} & \meter\\
\lstick{\rho_{1}} &/\qw & \multigate{1}{e^{-i\mathcal{S}\Delta t}} & \qw & \qw & \multigate{1}{e^{-i\mathcal{S}\Delta t}} & \qw & \qw & \multigate{1}{e^{-i\mathcal{S}\Delta t}} & \qw & \ldots & & \qw& \multigate{1}{e^{-i\mathcal{S}\Delta t}}&\qw\\
\lstick{\rho_{2}} &/\qw & \ghost{e^{-i\mathcal{S}\Delta t}} & \qw  & \lstick{\rho_{3}} & \ghost{e^{-i\mathcal{S}\Delta t}}  & \qw & \lstick{ \rho_{4}} & \ghost{e^{-i\mathcal{S}\Delta t}} & \qw & \ldots & & \lstick{\rho_{Q+1}} & \ghost{e^{-i\mathcal{S}\Delta t}} & \qw
}
\]
\caption{\footnotesize A quantum circuit for estimating $\tr(\rho\cos(\rho t))$ using qubit reset. The break and a state $\rho$ in the wire means implementing qubit reset.}
\label{figure:circuit_instance_reset}
\end{figure*}

\subsection{Subroutine: Swap operator exponentiation}\label{sec:swap}
In this section, we devise a circuit to simulate the exponentiation of the Swap operator by the technique for simulating a linear combination of unitaries in \cite{Berry2015}. Note that the unitary $e^{-i\mathcal{S}\Delta t}$ can be written as a linear combination of unitaries, i.e., $e^{-i\mathcal{S}\Delta t}=\cos(\Delta t)I_{2n}-i\sin(\Delta t)\mathcal{S}$. The index of identity means the number of qubits that the identity acts on. To break down the exponentiation of the Swap operator, we need to build up the module $W$ first, which is shown in Figure~\ref{figure:circuit_W}. The first two qubits in state $\ket{00}$ are newly added ancillary qubits.

\begin{figure}[htb]
\[ 
\Qcircuit @C=0.6em @R=0.8em {
\lstick{\op{0}{0}} & \gate{R_{1}} & \ctrlo{1} & \qw & \qw & \qw \\
\lstick{\op{0}{0}} & \qw & \gate{R_{2}} & \multigate{2}{{\rm select(\mathcal{S})}} & \gate{R_{2}^{\dagger}} & \qw \\
\lstick{\rho} &/\qw & \qw & \ghost{{\rm select(\mathcal{S})}} & \qw & \qw \\
\lstick{\rho} &/\qw & \qw & \ghost{{\rm select(\mathcal{S})}} & \qw & \qw^{\quad\quad W}
\gategroup{1}{2}{4}{5}{2.1em}{--}
}
\]
\caption{\footnotesize Quantum circuit for implementing the module $W$.}
\label{figure:circuit_W}
\end{figure}
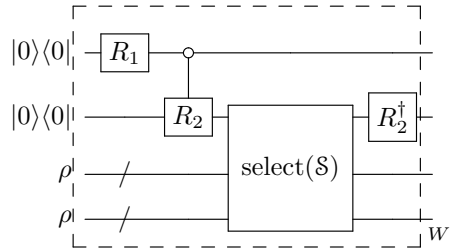

First, the $R_1$ gate is a rotation gate, the oc-$R_2$ gate means that act the rotation gate $R_2$ on the target qubit when the control qubit is in the state $\ket{0}$. The two single qubit rotations $R_1$ and $R_2$ are defined as follows: 

\begin{align}
    &R_{1}\ket{0}=\frac{\alpha}{2}\ket{0}+\sqrt{1-\frac{\alpha^2}{4}}\ket{1},\label{eq:rot_1}\\
    &R_{2}\ket{0}=\sqrt{\frac{\cos(\Delta t)}{\alpha}}\ket{0}+\sqrt{\frac{|\sin(\Delta t)|}{\alpha}}\ket{1},\label{eq:rot_2}
\end{align}
where $\alpha=\cos(\Delta t)+|\sin(\Delta t)|<2$. Here, we assume time $\Delta t$ is small enough such that $\cos(\Delta t)>0$.

Second, the select($\mathcal{S}$) gate implements the operation $(-i\cdot{\rm sgn})\mathcal{S}$ on $\rho\otimes\rho$ conditionally, which is defined as follows:
\begin{align}
    {\rm select}(\mathcal{S})=\op{0}{0}\otimes I_{2n}+\op{1}{1}\otimes (-i\cdot{\rm sgn})\mathcal{S},
\end{align}
where ${\rm sgn}$ denotes the sign of $\Delta t$. The detailed structure can be found in Appendix \ref{sec:circuit_decomposition}.

Third, we define a circuit module $W$ as shown in Figure~\ref{figure:circuit_W}. At this stage, the module can be written as:
\begin{align}
W=R_{2}^{\dagger}{\rm select}(\mathcal{S})({\rm oc-}R_{2})R_{1}. \label{eq:W_def}
\end{align}

Let $P=\op{00}{00}$ be the operator that projects onto the subspace spanned by $\ket{00}$, where $\ket{00}$ are the ancillary qubits in Figure~\ref{figure:circuit_W}. Then, define a unitary operator

\begin{align}
    A=-W(I_{2}-2P)W^{\dagger}(I_{2}-2P)W. \label{eq:circuit_A}
\end{align}
Here, the notation $I_2-2P$ denotes the operator that reflects along the vectors that are orthogonal to the ancillary qubits $\ket{00}$. As a result, the unitary $A$ can simulate $e^{-i\mathcal{S}\Delta t}$.

\begin{proposition}\label{le:swap}
For arbitrary parameter $\Delta t\in(-1,1)$, define two rotations $R_1$ and $R_2$ as in Eqs.~\eqref{eq:rot_1}-\eqref{eq:rot_2}. Define a circuit module $W$ as in Figure \ref{figure:circuit_W} and a unitary $A=-W(I_{2}-2P)W^{\dagger}(I_{2}-2P)W$, where $P=\op{00}{00}$, and $I_{2}$ denotes the identity acting on $\ket{00}$. Then the unitary $e^{-i\mathcal{S}\Delta t}$ can be simulated in the sense that
\begin{align}
    P\otimes e^{-i\mathcal{S}\Delta t}=PAP.
\end{align}
\end{proposition}
The proof is provided in Appendix~\ref{app:swap}.

\begin{figure*}[hbt]
\[ 
\Qcircuit @C=0.2em @R=0.7em {
\lstick{\op{0}{0}}  & \gate{H} & \ctrl{1}  & \qw & \ctrl{2} & \qw & \qw & \qw & \ctrl{1} & \ctrl{1} & \qw & \qw & \ctrl{2} & \qw & \ctrl{1} & \ctrl{1} & \ctrl{1} & \ctrl{1} & \qw  & \ctrl{2} & \qw & \gate{H} & \meter \\
\lstick{\op{0}{0}}  & \qw      & \gate{R_{1}} & \ctrlo{1} & \qw & \qw & \qw & \gate{X} &\gate{Z} & \ctrl{1} & \gate{X} & \qw &\qw & \ctrlo{1} & \gate{R_{1}^{\dagger}} & \gate{Z}  & \ctrlo{1} & \gate{R_{1}} & \ctrlo{1} & \qw & \qw & \qw & \qw \\
\lstick{\op{0}{0}}  & \qw      & \qw & \gate{R_{2}} & \multigate{2}{{\rm select(\mathcal{S})}} & \qw &  \gate{R_{2}^{\dagger}} & \qw & \qw &\gate{Z} & \qw & \gate{R_{2}} & \multigate{2}{{\rm select(\mathcal{S})}^{\dagger}} & \gate{R_{2}^{\dagger}} & \qw & \qw & \gate{Z} & \qw &\gate{R_{2}} & \multigate{2}{{\rm select(\mathcal{S})}} & \gate{R_{2}^{\dagger}} & \qw  &\qw \\
\lstick{\rho} &/\qw & \qw & \qw & \ghost{{\rm select(\mathcal{S})}} & \qw & \qw & \qw & \qw & \qw & \qw & \qw & \ghost{{\rm select(\mathcal{S})}^{\dagger}} & \qw & \qw & \qw & \qw & \qw & \qw & \ghost{{\rm select(\mathcal{S})}} & \qw & \qw & \qw \\
\lstick{\rho} &/\qw & \qw  & \qw& \ghost{{\rm select(\mathcal{S})}} & \qw & \qw & \qw & \qw^{ controlled-W}
\gategroup{1}{3}{5}{7}{1.5em}{.} & \qw & \qw & \qw & \ghost{{\rm select(\mathcal{S})}^{\dagger}} & \qw & \qw^{ controlled-W^\dagger}\gategroup{1}{12}{5}{15}{1.5em}{.} & \qw  & \qw & \qw & \qw^{ controlled-W} & \ghost{{\rm select(\mathcal{S})}} & \qw & \qw^{}\gategroup{1}{18}{5}{21}{1.5em}{.} & \qw^{ controlled-A}\gategroup{1}{3}{5}{21}{2.2em}{--}
}
\]
\caption{\footnotesize 
This figure depicts the resulting circuit by substituting c-$e^{-i\mathcal{S}\Delta t}$ with the circuit of controlled-$A$ (dashed box) in Figure~\ref{figure:circuit_exponentiation}. The dotted circuit is the controlled-$W$ circuit, in which the c-$R_1$ and oc-$R_2$ are the 1-controlled $R_1$ gate (apply $R_1$ on the target qubit if the control qubit in state $\ket{1}$) and 0-controlled $R_2$ gate (apply $R_2$ on the target qubit if the control qubit in state $\ket{0}$), respectively. The definitions of $R_1$ \& $R_2$ can be found in Eqs.~\eqref{eq:rot_1}-\eqref{eq:rot_2}.  The circuits between dotted boxes are known as reflectors. Denote that all elements in the circuit can be broken down into single/two-qubits gates, please refer to Appendix~\ref{sec:circuit_decomposition} for details.
}
\label{figure:circuit_A}
\end{figure*}
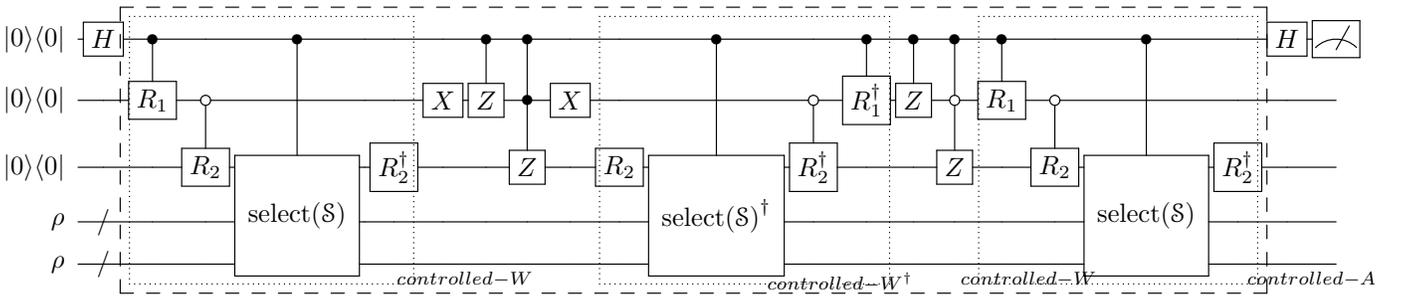

Using the circuit of $A$ to substitute $e^{-i\mathcal{S}\Delta t}$ would lead to the desired quantum circuits for estimating the Fourier series. We provide one example in Figure~\ref{figure:circuit_A}, where the circuit of $A$ is used once. We also estimate the number of needed primitive gates and qubits in the following.
\begin{proposition}\label{th:term}
For any quantum state $\rho\in\mathbb{C}^{2^n\times 2^n}$, and time $t\in\mathbb{R}$, there is a quantum circuit to estimate the quantity $\tr(\rho\cos(\rho t))$ up to precision $\epsilon$. The total amount of single/two-qubit gates is $O(nt^2/\epsilon)$.
\end{proposition}
\begin{proof}
The validity follows immediately from Proposition~\ref{th:main_cost} and Proposition~\ref{le:swap} and the gate decomposition of controlled $e^{-i\mathcal{S}\Delta t}$ in Appendix~\ref{sec:circuit_decomposition}.
\end{proof}

By Proposition~\ref{th:term}, we can estimate $\tr(\rho\cos(\rho t))$ using $O(t^2/\epsilon)$ primitive gates. In consequence, we can estimate the Fourier series by primitive gates as well. 


\section{Quantum algorithms}\label{sec:quantum_algorithm}
In this section, we present the quantum algorithms for estimating the von Neumann and $\alpha$-R\'enyi entropy. The key idea is to run the devised quantum circuits to evaluate the Fourier series in Lemma \ref{le:fourier} (for von Neumann entropy) and Lemma \ref{le:renyi_fourier} (for R\'enyi entropy). To be more specific, we construct an unbiased estimator by evaluating the Fourier series via the circuits and classical post-processing the measurement outcomes. In particular, in post-processing, we use the importance sampling technique. Moreover, we analyze the correctness and cost of our algorithms. 


\begin{algorithm}[htb]
\footnotesize
\caption{Quantum algorithm for von Neumann entropy estimation}
\label{alg:entropy_estimation}
\begin{algorithmic}[1]\label{algorithm: algorithm 1}
    \REQUIRE Constants $\epsilon,\delta,\Lambda\in(0,1)$, and copies of state $\rho\in\mathbb{C}^{2^n\times2^n}$.
    \ENSURE  Estimate of the von Neumann entropy $S(\rho)$.
\STATE Compute coefficients $L$, $K$, $M_l$, and $b_{l}^{(k)}$ as given in Lemma~\ref{le:fourier}.
\STATE Compute $\|\mathbf{f}\|_{\ell_1}$, as given in Eq.~\eqref{eq:weights}.
\STATE Set estimation error $\varepsilon=\epsilon/\|\mathbf{f}\|_{\ell_1}$.
\STATE Set integer $N=\sum_{l=0}^{\lfloor L\rfloor}(2M_l+1)$.
\STATE Define a distribution as in Eq.~\eqref{eq:random_variable}.
\STATE Set integer $B=\#{\rm Samples}$ as in Eq.~\eqref{eq:samples}.
\STATE Sample $B$ pairs of $(s_1,l_1)$, $\ldots$, $(s_B,l_B)$. 
\STATE Set $j=1$ and ${\rm sum}=0$.
\WHILE{$j\leq B$}
\STATE Set $Q=\lceil (2s_j-l_j)^2\pi^2/4\varepsilon\rceil$.
\STATE Prepare $Q+1$ copies of $\rho$.
\STATE Estimate $\tr(\rho\cos(\rho (2s_j-lj))\pi/2)$ with precision $\varepsilon$ and probability $1-\delta/2N$ via quantum circuits, given in Sec.~\ref{sec:quantum_circuit}.
\STATE Store the obtained estimate ${\rm est}_{j}$.
\STATE Update ${\rm sum}\leftarrow {\rm sum}+ {\rm est}_{j}$ and $j\leftarrow j+1$.
\ENDWHILE
\RETURN{$\|\mathbf{f}\|_{\ell_1}\times{\rm sum}/B$.}
\end{algorithmic}
\end{algorithm}
\paragraph{Quantum algorithm for von Neumann entropy estimation}
The workflow for estimating $S(\rho)$ is depicted in Algorithm~\ref{alg:entropy_estimation}. First, we set a constant $\Lambda$ as a lower bound on all non-zero eigenvalues of the input state $\rho$. Then, upon receiving the required precision $\epsilon$ and failure probability $\delta$, we determine the Fourier series according to Lemma \ref{le:fourier}. After that, we construct an unbiased estimator by the importance sampling technique. Particularly, we randomly select each term of the Fourier series, where the probability is proportional to the corresponding weight. The average of the selected terms in expectation is proportional to the target quantity. And the proportional factor is easy to calculate. Finally, we evaluate all selected terms via quantum circuits and post-processing the measurement outcomes to reveal the desired estimates. 

\begin{theorem}\label{th:main_result_sampling}
Consider a quantum state $\rho\in\mathbb{C}^{2^n\times2^n}$. Let $\Lambda$ be the lower bound on all non-zero eigenvalues of $\rho$. Suppose we have free access to copies of $\rho$, then Algorithm~\ref{alg:entropy_estimation} outputs an estimate of the von Neumann entropy $S(\rho)$ up to precision $\epsilon$, succeeding with probability at least $1-\delta$. In addition, the total amount of the needed copies of $\rho$ and single/two-qubit gates are, in the worst case, $\widetilde{O}\left({1}/{\epsilon^5\Lambda^2}\right)$ and $\widetilde{O}(n/\epsilon^3\Lambda^2)$, respectively.\footnote{The notation $\widetilde{O}$ hides the logarithmic factors.}
\end{theorem}
\begin{proof}
\textbf{Correctness analysis}.

We rephrase the Fourier series in Lemma~\ref{le:fourier} as follows.
\begin{align}
    F(\rho)=\sum_{l=0}^{\lfloor L\rfloor}\sum_{s=D_l}^{U_l}f(s,l)\tr\left(\rho \cos\left(\rho\frac{(2s-l)\pi}{2}\right)\right).
\end{align}
Coefficients $f(s,l)$ are given by
\begin{align}
    f(s,l)=\left(\sum_{k=1}^{K}\frac{b_{l}^{(k)}}{k}\right)\frac{\binom{l}{s}}{2^{l}}, \quad\forall s,l. \label{eq:weights}
\end{align}
Let $\mathbf{f}$ be a vector that consists of $f(s,l)$. Then the $\ell_1$-norm of $\mathbf{f}$ is bounded by $O(\log(K))$ i.e., $\|\mathbf{f}\|_{\ell_1}\in O\left(\log(K)\right)$ (cf. Lemma~\ref{le:fourier}). 

Next, define an importance sampling as follows:
\begin{align}
    &\mathbf{R}=\tr\left(\rho \cos\left(\rho\frac{(2s-l)\pi}{2}\right)\right) \quad \text{with prob. $\frac{f(s,l)}{ \|\mathbf{f}\|_{\ell_1}}$}.\label{eq:random_variable}
\end{align}
The random variable $\mathbf{R}$ indicates that each Fourier term associated with $(s,l)$ is sampled with probability proportional to its weight $f(s,l)$. Then, the Fourier series could be rewritten as an expectation of $\mathbf{R}$. 
\begin{align}
    F(\rho)=\left\|\mathbf{f}\right\|_{\ell_1}\cdot\mathbf{E}\left[\mathbf{R}\right].
\end{align}
Hence, we derive an unbiased estimator $F(\rho)$ for $S(\rho)$.

\textbf{Cost analysis}.

Note that the sample mean could estimate the expectation, and the estimation accuracy replies on the number of samples. Specifically, by Chebyshev's inequality, $O({\bf Var}/\epsilon^2)$ samples are sufficient to derive an estimate of the expectation with precision $\epsilon$ and high probability, where ${\bf Var}$ denotes the variance, and $\epsilon$ is the precision. Meanwhile, the probability could be boosted to $1-\delta$ at the cost of an additional multiplicative factor $O(\log(1/\delta))$ according to Chernoff bounds. Alternatively, by Hoeffding's inequality, we only need $O(\log(1/\delta)/\epsilon^2)$ samples to derive an estimate with precision $\epsilon$ and probability larger than $1-\delta$.

Regarding random variable $\mathbf{R}$, the variance is less than $1$. To estimate $F(\rho)$ with precision $\epsilon$, we set the estimation precision for $\mathbf{E}[\mathbf{R}]$ as $\varepsilon=\epsilon/\left\|\mathbf{f}\right\|_{\ell_1}$ and set failure probability as $\delta/2$. As a result, the number of needed samples is 
\begin{align}
    \#\textrm{Samples}&= O\left(\frac{1}{\varepsilon^2}\log\left(\frac{2}{\delta}\right)\right)=O\left(\frac{\|\mathbf{f}\|_{\ell_1}^2}{\epsilon^2}\log\left(\frac{2}{\delta}\right)\right).\label{eq:samples}
\end{align}
It means that there are at most $\#{\rm Samples}$ terms needing to estimate via quantum circuits. 

On the other hand, notice that the Fourier series $F(\rho)$ consists of $N=\sum_{l=0}^{L}(2M_l+1)$ terms in all. Hence, it suffices to estimate each term with probability $1-\delta/2N$. In this way, the overall failure probability is at most $\delta$ by union bound.

When estimating the Fourier series, we need to measure the resultant state after evolving the input state $\rho$ by the circuits (e.g., please refer to Figures~\ref{figure:circuit_instance} \& \ref{figure:circuit_instance_reset} \& \ref{figure:circuit_A}). Note that the measurement outcome is evaluated nondeterministically, which implies that the estimation could fail. To suppress the failure probability to $\delta$, we suffice to evaluate each term with a failure probability at most $\delta/2N$. In consequence, for each term, the number of needed measurements is
\begin{align}
  \#{\rm Measurements}=  O\left(\frac{\log(2N/\delta)}{\varepsilon^2}\right).\label{eq:measurements}
\end{align}
Immediately, the total number of measurements for estimating the Fourier series is at most
\begin{align}
C_{m}&=\#{\rm Sample}\times \#{\rm Measurements}=O\left(\frac{\|\mathbf{f}\|_{\ell_1}^4}{\epsilon^4}\log\left(\frac{2}{\delta}\right)\log\left(\frac{N}{\delta}\right)\right)=\widetilde{O}\left(\frac{1}{\epsilon^4}\right).
\end{align}


Now, we consider the cost on the number of state $\rho$. Recall that, when estimating $\tr(\rho\cos(\rho t))$, performing once measurement costs $\widetilde{O}(t^2/\epsilon)$ copies of $\rho$ by Proposition~\ref{th:main_cost}. For the Fourier series $F(\rho)$, the largest time is  $O(M_L)=O\left(\ln\left(\frac{\ln(K)}{\epsilon}\right)\frac{1}{\Lambda}\right)$. At the same time, we have to measure $C_m$ times in the entropy estimation. Then the number of overall copies is at most
\begin{align}
     C_{\rho}=C_{m}\times \widetilde{O}(M_L^2/\varepsilon) =C_{m}\times \widetilde{O}\left(\ln^2\left(\frac{\ln(K)}{\epsilon}\right)\frac{1}{\varepsilon\Lambda^2}\right)
     =\widetilde{O}\left(\frac{1}{\epsilon^5\Lambda^2}\right).
\end{align}

At last, consider the costs on the number of quantum gates. By Proposition~\ref{th:term}, for $\tr(\rho\cos(\rho t))$, it costs $O(nt^2/\epsilon)$ primitive single/two-qubit gates to construct the circuit. Meanwhile, we consider the largest evolution time $O(M_L)$. Then the total amount of gate for the entropy estimation, in the worst case, is
\begin{align}
    C_{g}=\# {\rm Sample} \times \widetilde{O}(nM_L^2/\varepsilon)=\widetilde{O}\left(\frac{n}{\epsilon^3\Lambda^2}\right).
\end{align}
\end{proof}

\paragraph{Quantum algorithm for quantum R\'enyi entropy estimation}
The quantum algorithm for the R\'enyi entropy is similar to that of the von Neumann entropy estimation. The main difference is that we evaluate the Fourier series approximation of $\tr(\rho^\alpha)$. Now, we show the workflow in Algorithm \ref{alg:renyi_entropy}, which derives an unbiased estimator for $\tr(\rho^\alpha)$ and an estimate for $R_{\alpha}(\rho)$. For clarity, we use $\epsilon$ and $\xi$ to denote the estimation precision of $R_{\alpha}(\rho)$ and $\tr(\rho^\alpha)$, respectively.

\begin{algorithm}[htb]
\footnotesize
\caption{Quantum algorithm for $\alpha$-R\'enyi entropy estimation}
\label{alg:renyi_entropy}
\begin{algorithmic}[1]
    \REQUIRE Constants $\epsilon,\delta,\Lambda\in(0,1)$, and copies of state $\rho\in\mathbb{C}^{2^n\times2^n}$, and $\alpha\in(0,1)\cup(1,\infty)$.
    \ENSURE Estimate of the $\alpha$-R\'enyi entropy $R_{\alpha}(\rho)$.
\STATE Set estimation precision $\xi$ for estimating $\tr(\rho^\alpha)$, as given in Eq. \eqref{eq:renyi_precision}.
\STATE Compute coefficients $L$, $K$, $M_l$, and $b_{l}^{(k)}$ as given in Lemma~\ref{le:renyi_fourier}.
\STATE Compute coefficients $f(s,l)=\left(\sum_{k=1}^{K}b_{l}^{(k)}(-1)^k\binom{\beta}{k}\right)\frac{\binom{l}{s}}{2^{l}}, \quad\forall s,l$.
\STATE Compute the $\ell_1$ norm of $\mathbf{f}$, which is the vector consisting of all $f(s,l)$. 
\STATE Set a parameter $\varepsilon=\xi/\|\mathbf{f}\|_{\ell_1}$.
\STATE Set integer $N=\sum_{l=0}^{\lfloor L\rfloor}(2M_l+1)$.
\STATE Define a distribution: $\mathbf{R}=\tr\left(\rho \cos\left(\rho\frac{(2s-l)\pi}{2}\right)\right) \quad \text{with prob. $\frac{|f(s,l)|}{ \|\mathbf{f}\|_{\ell_1}}$}.$
\STATE Set integer $B=\#{\rm Samples}=O\left(\frac{\|\mathbf{f}\|_{\ell_1}^2}{\xi^2}\log\left(\frac{2}{\delta}\right)\right)$.
\STATE Sample $B$ pairs of $(s_1,l_1)$, $\ldots$, $(s_B,l_B)$ via the distribution.
\STATE Set $j=1$ and ${\rm sum}=0$.
\WHILE{$j\leq B$}
\STATE Set $Q=\lceil (2s_j-l_j)^2\pi^2/4\varepsilon\rceil$.
\STATE Prepare $Q+1$ copies of $\rho$.
\STATE Estimate $\tr(\rho\cos(\rho (2s_j-lj))\pi/2)$ with precision $\varepsilon$ and probability $1-\delta/2N$ via quantum circuits, given in Sec.~\ref{sec:quantum_circuit}.
\STATE Store the obtained estimate ${\rm est}_{j}$.
\STATE Update ${\rm sum}\leftarrow {\rm sum}+ {\rm est}_{j}$ and $j\leftarrow j+1$.
\ENDWHILE
\RETURN{${\log(1+\|\mathbf{f}\|_{\ell_1}\times {\rm sum}/B)}/{(1-\alpha)}$.}
\end{algorithmic}
\end{algorithm}


In the workflow, first, upon receiving the inputs, we find a Fourier series that can approximate the trace of the state's power function $\tr(\rho^\alpha)$ by Lemma \ref{le:renyi_fourier}. Then we proceed to evaluate the Fourier series by the quantum circuits and the classical post-processing. Note that we also employ the importance sampling to construct an unbiased estimator. Afterwards, we could readily derive the estimate of $R_{\alpha}(\rho)$. 

As the analysis of Algorithm~\ref{alg:renyi_entropy} is similar to Algorithm~\ref{alg:entropy_estimation}, we provide the results in Theorem~\ref{th:main_renyi_sampling} and defer the details to Appendix~\ref{app:renyi_algorithm}.



\begin{theorem}\label{th:main_renyi_sampling}
Consider a quantum state $\rho\in\mathbb{C}^{2^n\times2^n}$. Let $\Lambda$ be the lower bound on all non-zero eigenvalues of $\rho$. Suppose we have access to copies of $\rho$, then Algorithm~\ref{alg:renyi_entropy} outputs an estimate of $\alpha$-R\'enyi entropy $R_{\alpha}(\rho)$ up to precision $\epsilon$, succeeding with probability at least $1-\delta$. Furthermore, the total amount of the needed copies of $\rho$ and single/two-qubit gates, in the worst case, are shown in the Table~\ref{tab:renyi}.
\begin{table}[htb]
    \centering
    \begin{tabular}{|c|c|c|}
    \hline
        $\alpha$ & Copy cost & Gate cost \\
        \hline
        $(0,1)\cup(2,+\infty)$ & $\widetilde{O}\left(\frac{\left(\sum_{k=1}^{K}|\binom{\alpha-1}{k}|\right)^5}{|1-\alpha|^5[\tr(\rho^2)]^{5(\alpha-1)}(\epsilon)^5\Lambda^2}\right)$ & $\widetilde{O}\left(\frac{n\left(\sum_{k=1}^{K}|\binom{\alpha-1}{k}|\right)^3}{|1-\alpha|^3[\tr(\rho^2)]^{3(\alpha-1)}(\epsilon)^3\Lambda^2}\right)$ \\[0.3cm]
        \hline
        $(1,2]$ & $\widetilde{O}\left(\frac{\left(\sum_{k=1}^{K}|\binom{\alpha-1}{k}|\right)^5}{|1-\alpha|^5[\tr(\rho^2)]^5(\epsilon)^5\Lambda^2}\right)$ & $\widetilde{O}\left(\frac{3\left(\sum_{k=1}^{K}|\binom{\alpha-1}{k}|\right)^3}{|1-\alpha|^3[\tr(\rho^2)]^3(\epsilon)^3\Lambda^2}\right)$ \\[0.3cm]
        \hline
    \end{tabular}
    \caption{Cost estimation of Algorithm~\ref{alg:renyi_entropy}}
    \label{tab:renyi}
\end{table}
\end{theorem}

\begin{remark}
From results in Theorems \ref{th:main_result_sampling} \& \ref{th:main_renyi_sampling}, we can easily see that the costs of our algorithms scale polynomially in $1/\Lambda$, where $\Lambda$ is the lower bound on all non-zero eigenvalues. If $\Lambda$ is polynomially tiny, i.e., $\Lambda=\Omega(1/poly(n))$, the number of needed quantum states and primitive gates could be polynomially large. Hence, our algorithms can apply to large states once the states satisfy the aforementioned assumption. For example, it may happen when a large state has a low rank.
\end{remark}


\section{Numerical simulation}\label{sec:experiment}
In this section, we conduct the numerical simulation to demonstrate the effectiveness and correctness of our algorithms. Specifically, we estimate $S(\rho)$ and $R_{2}(\rho)$ for several randomly generated single-qubit state $\rho$. 
And, we estimate the entropy of a single-qubit state under the effect of depolarizing noises and amplitude damping noises.
All simulation experiments are operated on {\textit{Paddle Quantum}} platform. 


\subsection{Effectiveness and correctness}\label{subsec:Correctness}
To begin with, we generate a single-qubit mixed state at random and denote it by $\rho= 
    \begin{pmatrix}
        0.48786 & 0.0094\\
        0.0094 & 0.51214
    \end{pmatrix}$.
As the eigenvalues of $\rho$ are larger than $0.35$, we set the eigenvalue lower bound as $\Lambda=0.35$ (as long as smaller than the minimum non-zero eigenvalue of the state $\rho$). In the experiment, we set the error tolerance $\epsilon=0.2, 0.4$, respectively. The results are depicted in Figure~\ref{fig:Entropy estimation by quantum circuit}.
\begin{figure}[htb]
\centering
\subfigure[ \ Estimated von Neumann entropy. \ ]{\includegraphics[width=0.47\textwidth]{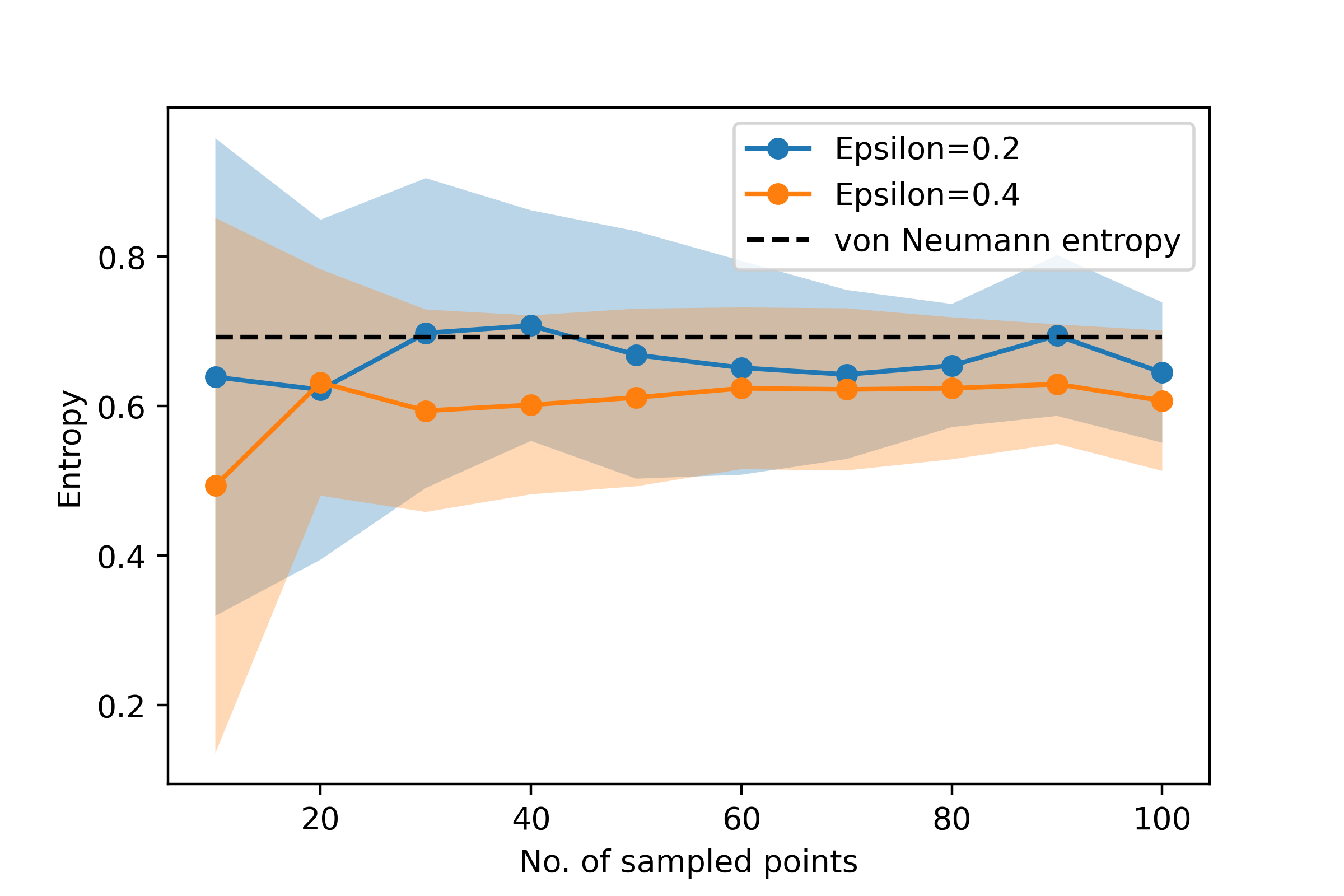}}
\subfigure[ \ Estimated 2-R\'enyi entropy. \  ]{\includegraphics[width=0.47\textwidth]{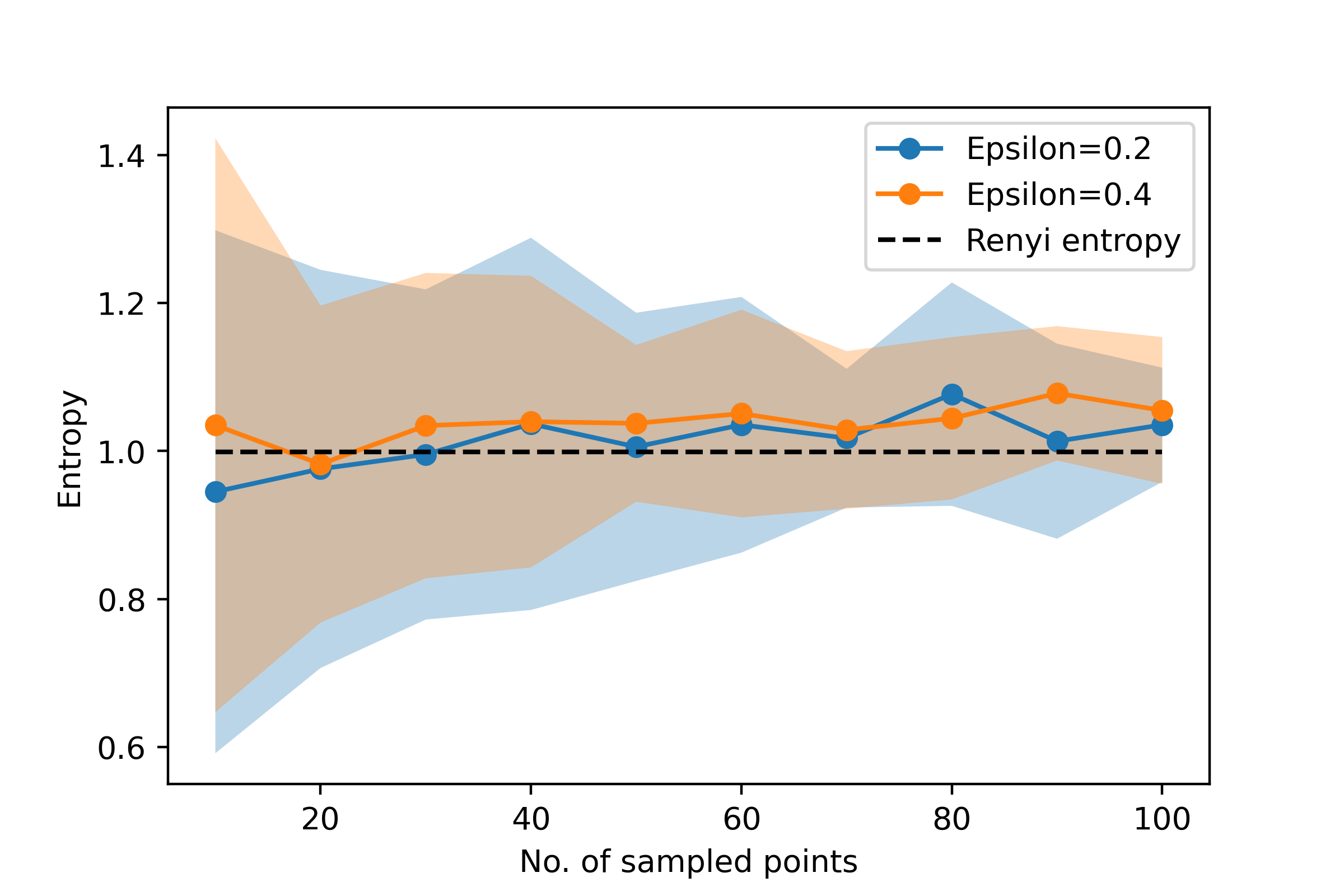}}
\caption{\footnotesize In (a) and (b), the black dashed line represents the actual entropy of quantum state $\rho$. The blue and orange curves are average entropy over 20 repeats for $\epsilon$ equal to 0.2 and 0.4, respectively. The shadowed area stands for standard deviation.}
\label{fig:Entropy estimation by quantum circuit}
\end{figure}

In Figure~\ref{fig:Entropy estimation by quantum circuit}, the coloured curves represent the estimates of the entropy with the error tolerance $0.2,0.4$. The shadowed areas represent standard deviation. Clearly, both blue and orange curves fluctuate around the black dashed line,
but the blue curve ($\epsilon=0.2$) is closer.
In addition to that, the shadowed areas converge as the number of sampled points increases, meaning the estimation gets precise. Hence, we could conclude that our method could estimate the entropy precisely with a large number of sampled points and small error tolerance $\epsilon$.

Next, we show the effectiveness with more quantum states. We randomly generate four more single-qubit mixed states to match the lower bound $\Lambda$ equal to 0.35 (These density matrices are displayed in Appendix \ref{sec:quantum state}). We set the number of sampled points as $100$ and error tolerance $\epsilon=0.2$. The corresponding results are illustrated in Figure~\ref{fig:histgram}.


\begin{figure}[htb]
\centering
\subfigure[ \ Estimated von Neumann entropy. \ ]{\includegraphics[width=0.45\textwidth]{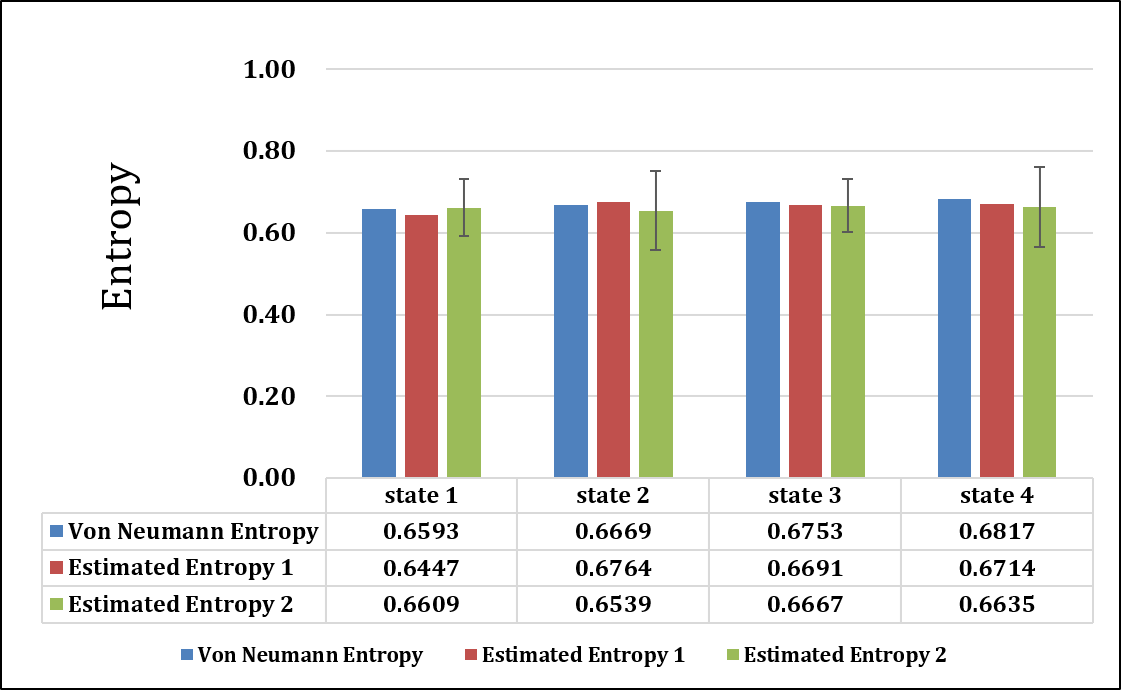}}
\subfigure[ \ Estimated R\'enyi entropy. \  ]{\includegraphics[width=0.45\textwidth]{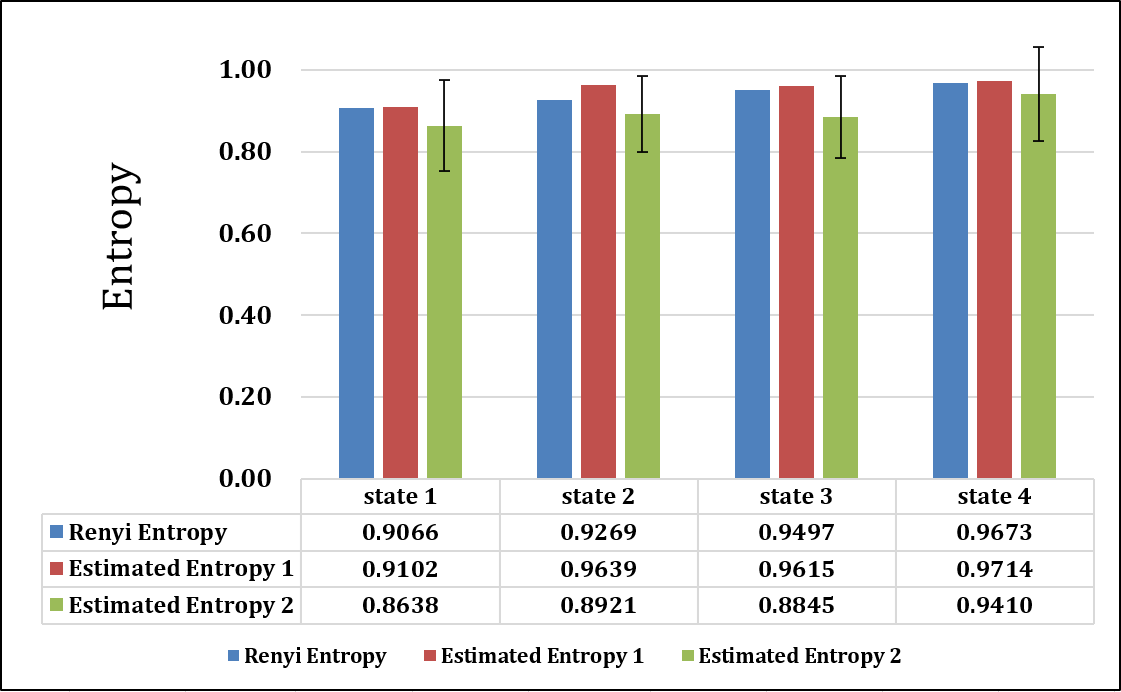}}
\caption{\footnotesize The results for 4 randomly generated states (available in Appendix \ref{sec:quantum state}). In (a) and (b), the blue bar is the real quantum entropy, the Estimated Entropy 1 stands for the entropy corresponding to the Fourier series approximation, and the Estimated Entropy 2 is the average entropy (100 sample points, repeat 20 times) calculated by our approach. In addition, the error bar represents the standard deviation. }
\label{fig:histgram}
\end{figure}


In Figure \ref{fig:histgram}, the blue bar is the true quantum entropy corresponding to $S(\rho)$ and $R_{2}(\rho)$. The red bar is the value calculated by the Fourier series (cf. Lemmas \ref{le:fourier} \& \ref{le:renyi_fourier}). 
The green bar denotes the estimate of the entropy by our approach. 
Clearly, the red and green bars are very close to the blue bar for different states. Thus, the experimental results show that our approach could find high-precision estimates for generated states, implying the validity of the Fourier series and our algorithms.

\subsection{Robustness}
We also study the performance of our approach under the effect of noises. Specifically, we consider a single qubit quantum state $\rho$ and single-qubit amplitude damping channel $\mathcal{N}^{amp}_p(\rho)$ and depolarizing channel $\mathcal{N}^{depl}_p(\rho)$. The noisy quantum channels are given by
\begin{align}
        \mathcal{N}^{amp}_p(\rho) &:= D_0\rho D_0^\dagger +  D_1\rho D_1^\dagger,\\
        \mathcal{N}^{depl}_p(\rho) &:= E_0\rho E_0^\dagger +  E_1\rho E_1^\dagger+E_2\rho E_2^\dagger +  E_3\rho E_3^\dagger,
\end{align} 
where 
\begin{align*}
D_{0}=\begin{bmatrix} 1 & 0 \\
		0 & \sqrt{1-p} \end{bmatrix}, D_1=\begin{bmatrix} 0 & \sqrt{p} \\
		0 & 0 \end{bmatrix},
		E_0=\sqrt{1-p}I,
	E_1=\sqrt{\frac{p}{3}}X,
	E_2=\sqrt{\frac{p}{3}}Y,
	E_3=\sqrt{\frac{p}{3}}Z.
\end{align*}
Here, the parameter $p\in [0,1]$ is the noise level. Notation $I$ stands for identity operator, and $X, Y, Z$ are Pauli matrices. 

In the experiment, the input the quantum state is $\rho = 
    \begin{pmatrix}
        0.5398 & -0.1217\\
        -0.1217 & 0.4602
    \end{pmatrix}$.
We set the noise level $p$ as 0, 0.02, 0.04, 0.06, 0.08, 0.10, 0.12, 0.15, respectively. And we set the error tolerance $\epsilon=0.2$ and eigenvalue lower bound $\Lambda=0.35$ and take 100 sample points. 
The results are shown in Figure \ref{fig:Entropy respect to noises}.

\begin{figure}[htb]
\centering
\subfigure[ \ $S(\rho)$ with amplitude damping noises. \ ]{\includegraphics[width=0.47\textwidth]{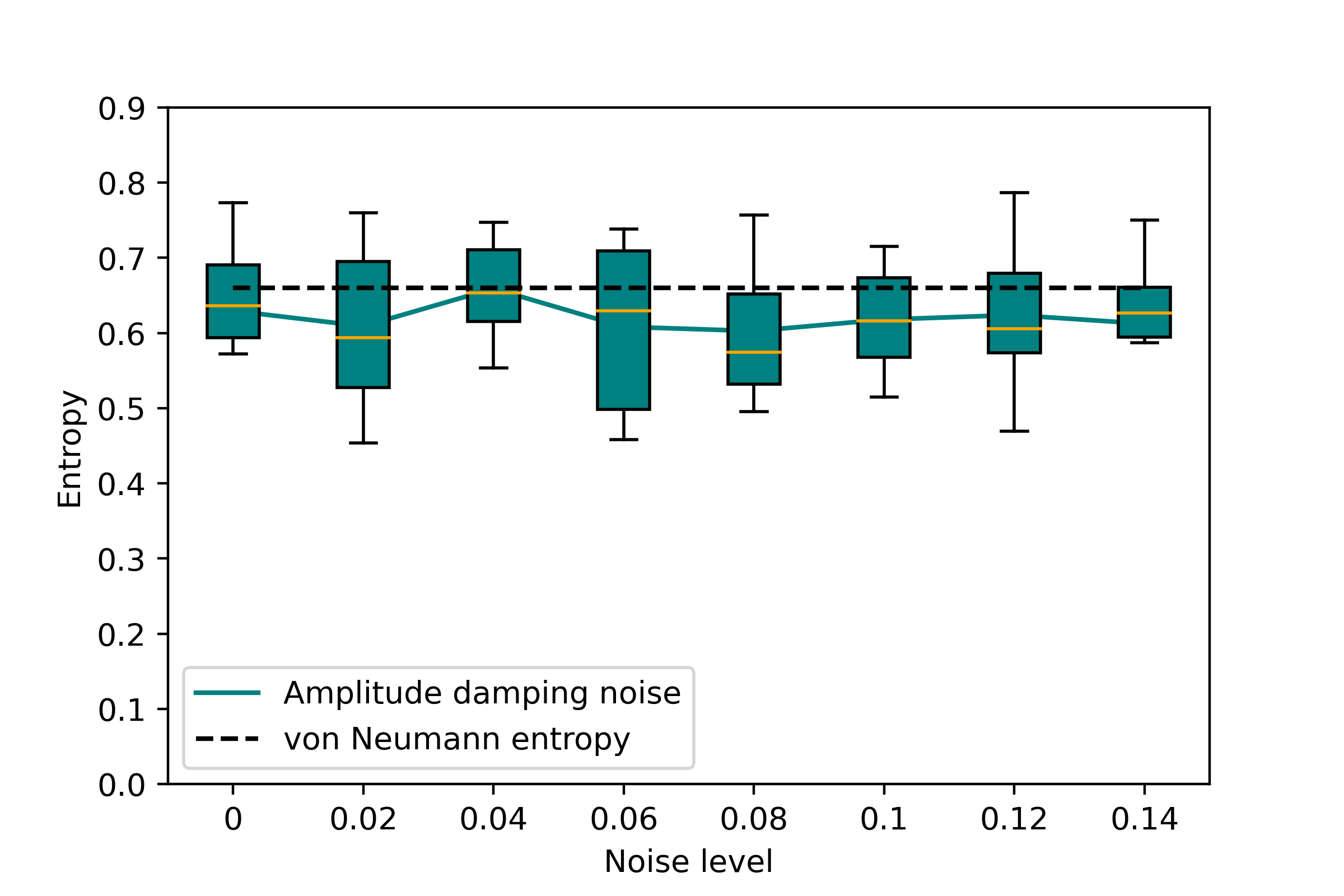}}
\subfigure[ \ $S(\rho)$ with depolarizing noises. \  ]{\includegraphics[width=0.47\textwidth]{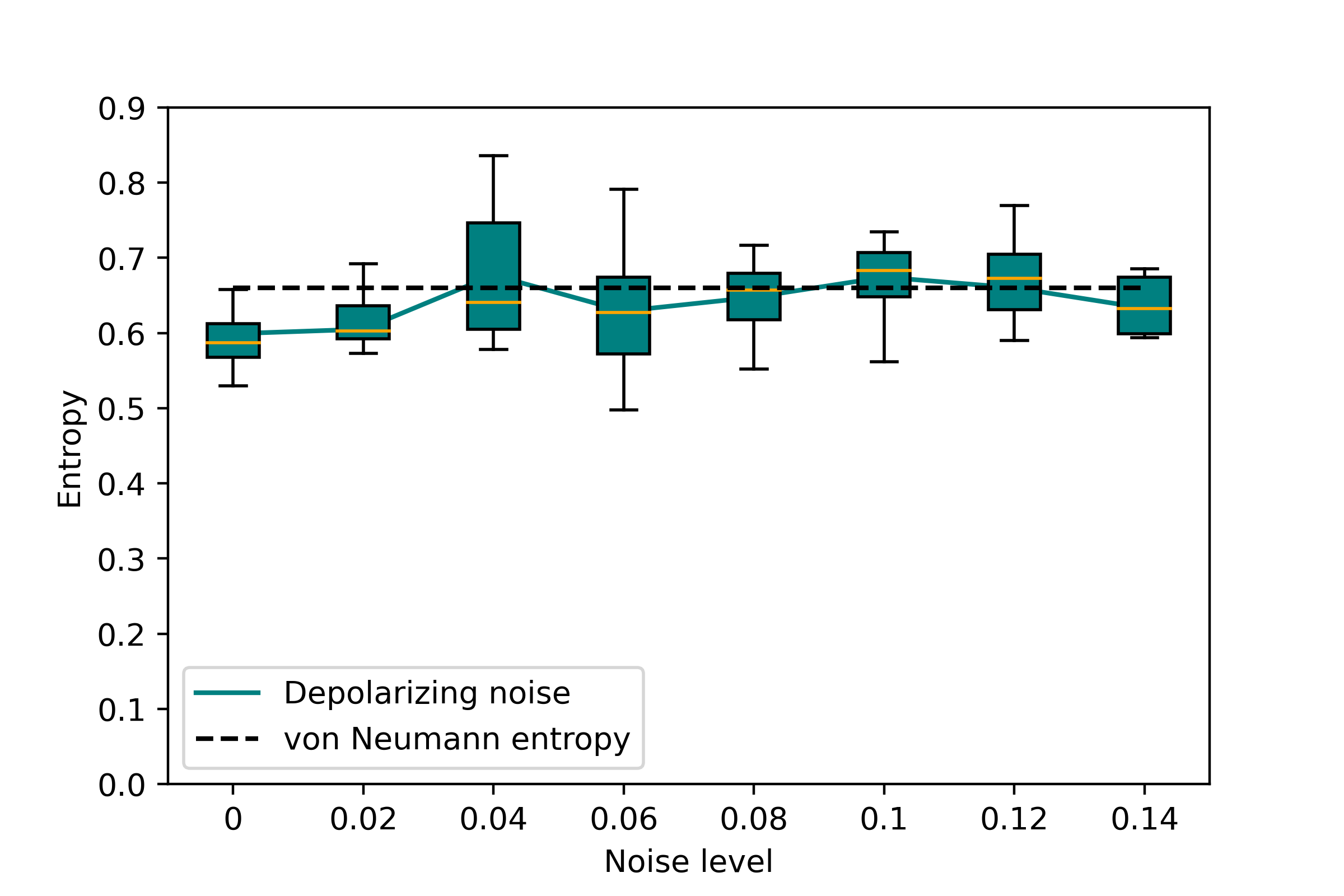}}
\subfigure[ \ $R_{2}(\rho)$ with amplitude damping noises. \ ]{\includegraphics[width=0.47\textwidth]{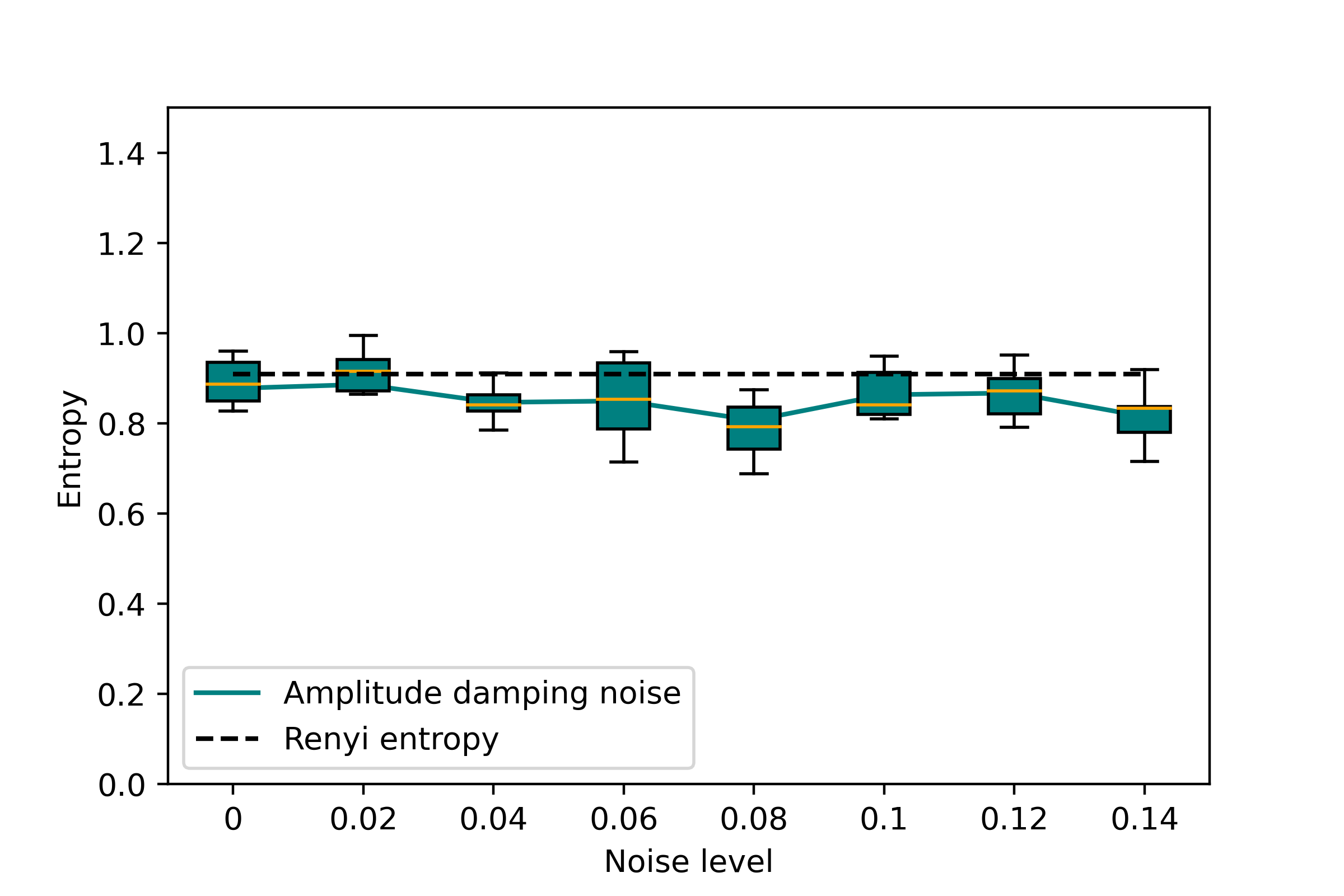}}
\subfigure[ \ $R_{2}(\rho)$ with depolarizing noises. \  ]{\includegraphics[width=0.47\textwidth]{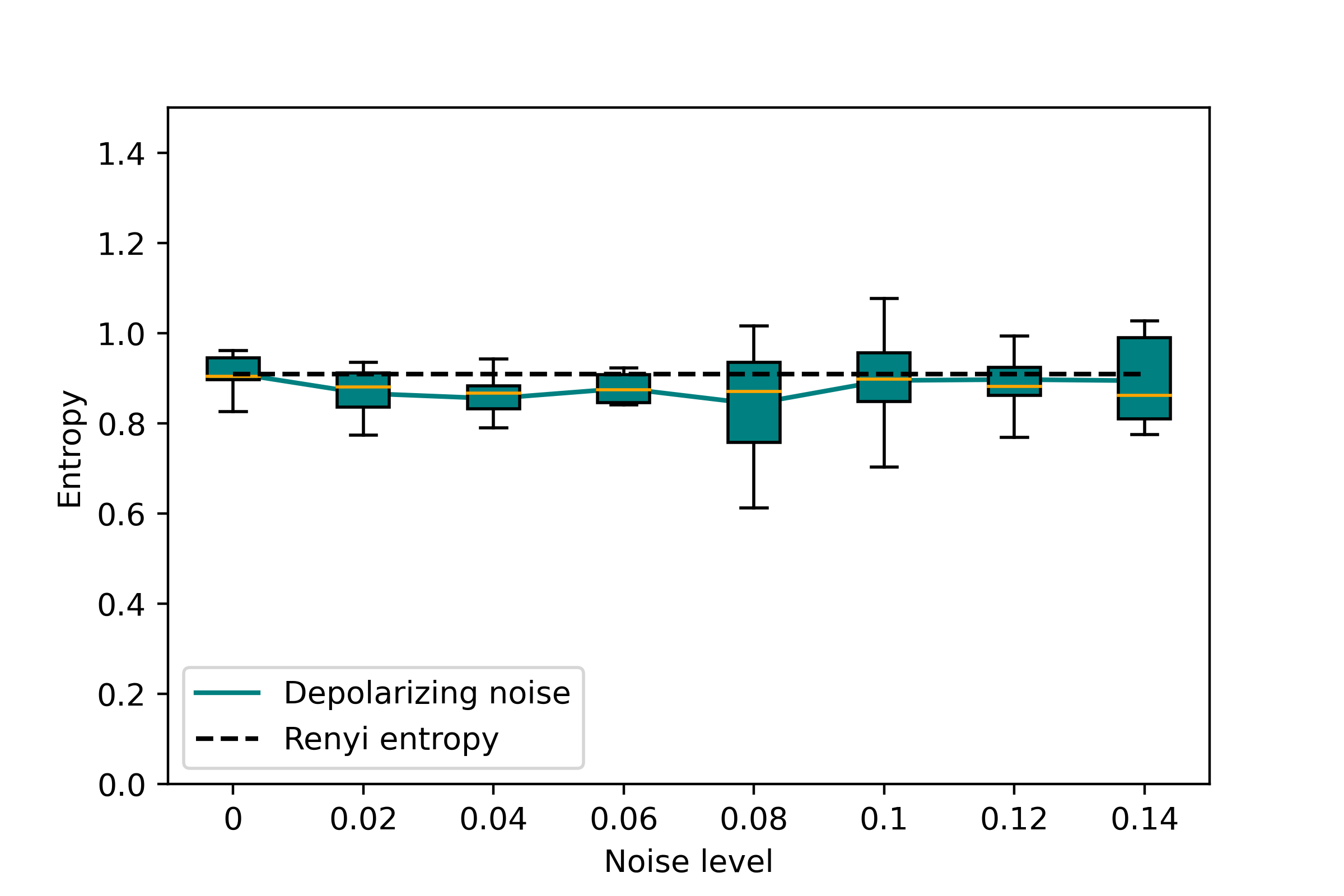}}
\label{fig:ad noise}
\caption{\footnotesize Figures (a) and (b) represent results for von Neumann entropy, and (c) and (d) represent the results for 2-R\'enyi entropy. The green curves link the average estimated entropy at different noise levels. The black dashed line represents the actual von Neumann entropy of quantum state $\rho$.  
}
\label{fig:Entropy respect to noises}
\end{figure}


In Figure \ref{fig:Entropy respect to noises}, (a) and (b) ((c) and (d)) are the box plots of $S(\rho)$ ($R_{2}(\rho)$) for amplitude damping noisy state and depolarizing noisy state, respectively. The black dashed line represents the true value of the entropy. The green curve represents the estimates of the entropy. 
As shown, all plots fluctuate around the black dashed line. Hence, we can confidently announce that our algorithm is robust to intermediate amplitude damping and depolarizing noises. 

In addition, the green curve in (b) and (d) is closer to the black dashed line than (a) and (c). Thus it implies that our algorithm performs better under depolarizing noises than amplitude damping noises for these chosen states.


\section{Discussions}\label{sec:Discussion}

\subsection{Comparison to literature}

The following discusses the differences between our results and previous related works~\cite{Acharya2020,Gilyen2019,Chowdhury2020,Luongo2020,Yirka2020,Gur2021}.




First, algorithms of \cite{Gilyen2019,Chowdhury2020,Luongo2020,Gur2021} require an oracle for preparing the purification of the input state. In other words, one has to find such quantum circuits for preparing the purification before performing the entropy estimation task. However, implementing the oracle for an unknown quantum state via quantum circuits is not easy in general. In contrast, our algorithm no longer require the oracle for implementing the circuits and only require access to the copies of the states, making our algorithms more practical than that of~\cite{Gilyen2019,Chowdhury2020,Luongo2020,Gur2021}. 

Second, when estimating the entropy of various states, one has to compile the circuit of the oracle for each input state to implement the quantum circuits of~\cite{Gilyen2019,Chowdhury2020,Luongo2020,Gur2021}, which means that one has to design different quantum circuits for various input states. This may cost additional resources for quantum compiling that maps the logical qubits of the algorithm to the physical qubits of the real quantum devices. In contrast, our algorithms can avoid this issue by using the same quantum circuits to estimate the entropy of various states. This is because the circuits in our algorithm are determined by the parameters $\Lambda$ and $\epsilon$, where $\Lambda$ is the lower bound of all non-zero eigenvalues of the input state, and $\epsilon$ is the precision (or called error tolerance). For a given precision and a given set of quantum states, our circuits can estimate the entropy of all states as long as the lower bound $\Lambda$ is set small enough. As a result, our approach would be more practical when dealing with multiple different quantum states.

Third, compared to \cite{Acharya2020,Gilyen2019,Luongo2020,Gur2021}, our algorithms could be more resource-efficient when the large state's minimal nonzero eigenvalue is polynomially small. The costs in previous works \cite{Gilyen2019,Luongo2020,Gur2021} are characterized in terms of the times of querying the oracle, depending on the dimension of the state. Meanwhile, 
the copy cost of \cite{Acharya2020} is exponentially large. In comparison, by Theorem \ref{th:main_result_sampling} \& \ref{th:main_renyi_sampling}, the copy costs of our algorithms scale polynomially with respect to the number of qubits $n$ when parameter $\Lambda=\Omega(1/poly(n))$.

Fourth, the approach in \cite{Yirka2020} can estimate the quantum R\'enyi entropy $R_{\alpha}(\alpha)$ when the parameter $\alpha$ is integer. In comparison, our approach is more general, i.e., our approach can apply to the case when $\alpha$ is an integer or non-integer.

\subsection{Applications}
Given the importance of quantum entropies, our algorithms will have various applications in science and engineering. 
Here, we discuss the applications in low-rank quantum states, quantum Gibbs state preparation and entanglement entropy estimation.
\paragraph{State preparation quality assessment.} 
The class of low-rank quantum states is particularly significant in physical experiments, forming a realistic model of quantum states prepared in the lab \cite{2015}. For instance, one important task is to prepare pure quantum states. In general, the prepared state tends to have rapidly decaying eigenvalues such that a low-rank state can well approximate it. If the low-rank state's quantum entropy  (disorder) can be quantified, we can assess the quality of the state preparation. Note that the generated state has a small number of significant eigenvalues, and the rests are close to zero. In this case, the generated state's eigenvalues with exponential scaling in the qubit counts can be ignored. Thus it is reasonable to some extent to assume that the minimal non-zero eigenvalue of the low-rank state is polynomially small in the worst case. Hence our algorithms can be applied to assess the quality of state preparation for this case.

\paragraph{Gibbs state preparation.}
Quantum Gibbs states or thermal states are of significant importance to quantum simulation~\cite{childs2018toward}, quantum machine learning~\cite{kieferova2017tomography,biamonte2017quantum}, and quantum optimizaiton~\cite{somma2008quantum}, etc. Several recent works \cite{Chowdhury2020,Wang2020} have proposed to use variational quantum algorithms to prepare the quantum Gibbs states. The core idea of these works is to train a parameterized quantum circuit (PQC) to generate a parameterized state $\rho(\bm\theta)$ matching the global minimum of the system's free energy. Specifically, the task is to minimize the loss function $L(\bm\theta)=\tr(H\rho(\bm\theta))-\beta^{-1}S(\rho(\bm\theta))$ by a gradient-descent method, where the notation $H$ denotes the system's Hamiltonian, and $\beta$ denotes the inverse temperature. One feasible scheme for estimating the gradient is the difference method, which demands efficient loss evaluation. As the loss evaluation involves the von Neumann entropy estimation, our algorithm could be employed for the gradient estimation. Besides, the PQC can be considered as a method to generate quantum state, which guarantees that we have the access to the target state freely. 


\paragraph{Entanglement entropy estimation}
Quantum entanglement~\cite{Horodecki2009a} plays a crucial role in quantum physics while
it is the central resource in many quantum information
applications such as teleportation, super-dense coding, and quantum key distribution~\cite{Ekert1991}. Thus, developing methods for entanglement quantification will be of great importance to the study in these fields. Notice that the entropy of entanglement~\cite{Bennett1996b} is one of the most important entanglement measures for quantifying bipartite entanglement since the rate of the transformation between any bipartite pure state $\ket{\psi}_{AB}$ and two-qubit singlet is given by its entropy of entanglement, i.e.,  von Neumann entropy of the subsystem $S(\tr_{B}(\op{\psi}{\psi}_{AB}))$. Hence, our approach can be employed to measure the entanglement via the von Neumann entropy estimation.

\subsection{Concluding remarks}\label{sec:conclusion}
This paper has provided quantum algorithms for estimating the von Neumann entropy and $\alpha$-R\'enyi entropy. The key of our algorithms is the quantum circuits that can efficiently evaluate the terms of the Fourier series and classically post-processing the measurement outcomes. Especially, quantum circuits are composed of primitive single/two-qubit gates. The design of quantum circuits synthesizes several quantum tools, such as iterative quantum phase estimation, the exponentiation of quantum state, qubit reset, and the linear combination of unitaries. 
The notable property of our algorithm is that the circuits do not use oracles, which makes our circuit more friendly to near-term quantum devices than the oracle-based quantum algorithms. Besides, our circuits are only determined by parameters $\Lambda$ and $\epsilon$, allowing the same quantum circuits to estimate multiple different quantum states. Furthermore, when the input state's minimal non-zero eigenvalue is polynomially scaling in the qubit counts, i.e., $\Lambda=\Omega(1/poly(n))$,
our algorithms only use polynomially many copies and thus have potential applications to large states processing. Given these merits, our algorithm may be expected to complete the entropy estimation task on the near and intermediate future quantum computing devices. 

It will be meaningful to further devise practical and efficient quantum algorithms to solve problems related to quantum entropies. These problems could be in the areas of optimization (combinatorial optimization problems~\cite{somma2008quantum}, semi-definite programming~\cite{Brandao2017}), quantum machine learning, many-body physics, and quantum chemistry. We also hope our results would shed light in the applications of near-term quantum devices in the study of condensed matter physics, high-energy physics, and gravity and black holes theory~\cite{dong2016gravity,bekenstein2020black}.


\vspace{0.5cm}

\textbf{Acknowledgments.---}
Y. W. acknowledges support from the Baidu-UTS AI Meets Quantum project and the Australian Research Council (Grant No: DP180100691). This work was done when B. Z. was a research interns at Baidu Research.

\bibliography{Reference}
\bibliographystyle{apsrev4-1}
\onecolumngrid

\newpage
\begin{center}
{\textbf{\large Appendix }}
\end{center}

\renewcommand{\theequation}{S\arabic{equation}}
\renewcommand{\thealgorithm}{S\arabic{algorithm}}
\setcounter{equation}{0}
\setcounter{figure}{0}
\setcounter{table}{0}
\setcounter{section}{0}
\setcounter{proposition}{0}
\setcounter{algorithm}{0}

\section{Fourier series}\label{app:fourier}
\subsection{Truncated Taylor series for von Neumann entropy}\label{sec:series}
\renewcommand{\thelemma}{S\arabic{lemma}}
\setcounter{lemma}{0}
\begin{lemma}\label{le:series}
Suppose the minimal non-zero eigenvalue of state $\rho$ is at least $\Lambda$, then there exists an integer $K$ such that 
\begin{align}
    \left|S(\rho)-\sum_{k=1}^{K}\frac{1}{k}\tr(\rho(I-\rho)^k)\right|\leq\epsilon,
\end{align}
where $K\in\Theta(\log(1/\epsilon\Lambda)/\log(1/(1-\Lambda))$, and $\sum_{k=1}^{K}\frac{1}{k}\leq \log(K+1)+1$.
\end{lemma}
\begin{proof}
First, we assume $\rho$ has a spectral decomposition $\rho=\sum_{j}\lambda_j\op{\mathbf{e}_j}{\mathbf{e}_j}$. Then the truncation error is given by
\begin{align}
    {\rm error}&=\left|\sum_{k=K+1}^{\infty}\frac{1}{k}\tr\left(\rho(I-\rho)^k\right)\right|\\
    &=\sum_{k=K+1}^{\infty}\sum_{j}\frac{1}{k}\lambda_j(1-\lambda_j)^k\\
    &\leq\sum_{k=K+1}^{\infty}\frac{1}{k}(1-\Lambda)^k\\
    &\leq \frac{(1-\Lambda)^{K+1}}{\Lambda(K+1)}.
\end{align}
Here we have used facts that $\sum_j\lambda_j=1$ and $\lambda_j\geq \Lambda$.

To suppress the truncation error to $\epsilon$, we choose $K\in O(\log(1/\epsilon\Lambda)/\log(1/(1-\Lambda)))$. This can be easily verified. As we aim to find $K$ such that $(1-\Lambda)^{K+1}/\Lambda(K+1)\leq\epsilon$, we take logarithm on both sides. Then we have
\begin{align}
    (K+1)\log(1-\Lambda)-\log(K+1)\leq \log(\epsilon\Lambda).
\end{align}
Multiply $-1$ on both sides.
\begin{align}
    (K+1)\log(1/(1-\Lambda))+\log(K+1)\geq \log(1/\epsilon\Lambda).
\end{align}
Clearly, $\log(K+1)\geq0$. We then choose $K$ such that $(K+1)\log(1/(1-\Lambda))\log(4/\epsilon\Lambda)$, immediately leading to the desired result.

Finally, we derive an upper bound on the Harmonic number $\sum_{k=1}^{K}\frac{1}{k}$. Notice that 
\begin{align}
    \frac{1}{k}=\frac{1}{k}\cdot1\leq\int_{x=k}^{k+1}\frac{1}{x}{\rm d}x+1\cdot(\frac{1}{k}-\frac{1}{k+1}).
\end{align}
Here, $\frac{1}{k}\cdot 1$ denotes the area of a rectangle. Accordingly, the R.H.S denotes the area of two regions, one of which is determined by $1/x$, varying from $1/k$ to $1/k+1$, and the other is a rectangle with height $1/k-1/(k+1)$ and width 1. Thus,
\begin{align}
    \sum_{k=1}^{K}\frac{1}{k}&\leq \sum_{k=1}^{K}\left[\int_{x=k}^{k+1}\frac{1}{x}{\rm d}x+1\cdot(\frac{1}{k}-\frac{1}{k+1})\right]\\
    &=\int_{x=1}^{K+1}\frac{1}{x}{\rm d}x+\sum_{k=1}^{K}(\frac{1}{k}-\frac{1}{k+1})\\
    &\leq \ln(K+1)+1.
\end{align}
Finally, the proof is completed.
\end{proof}

\subsection{Proof for Lemma~\ref{le:fourier}}\label{sec:fourier}
\renewcommand{\thelemma}{\arabic{lemma}}
\setcounter{lemma}{0}
\begin{lemma}
For arbitrary quantum state $\rho\in\mathbb{C}^{2^n\times 2^n}$, let $\Lambda$ be the lower bound on all non-zero eigenvalues of $\rho$. There exists a Fourier series $S(\rho)_{est}$ such that $\left|S(\rho)-S(\rho)_{est}\right| \leq\epsilon$ for any $\epsilon\in(0,1)$, where
\begin{align}
    S(\rho)_{est} = \sum_{l=0}^{\lfloor L\rfloor}\sum_{s=D_l}^{U_l}\sum_{k=1}^{K}\frac{b_{l}^{(k)}\binom{l}{s}}{k2^{l}}\tr(\rho\cos(\rho\cdot t(s,l))).\label{app:eq:entropy estimation}
\end{align}
Particularly, coefficient $U_{l}=\min\{l,\lceil \frac{l}{2}\rceil+M_l\}$ and $D_l=\max\{0,\lfloor \frac{l}{2}\rfloor-M_l\}$, and $\binom{l}{s}$ denotes the binomial coefficient. Meanwhile, coefficients $t(s,l)=(2s-l)\pi/2$, and coefficients $K,L,M_l$ are given by 
\begin{align}
K\in\Theta\left(\frac{\log(\epsilon\Lambda)}{\log(1-\Lambda}\right), \quad L=\ln\left(\frac{4\sum_{k=1}^{K}1/k}{\epsilon}\right)\frac{1}{\Lambda^2}, \quad M_l=\left\lceil\sqrt{\ln\left(\frac{4\sum_{k=1}^{K}1/k}{\epsilon}\right)\frac{l}{2}}\right\rceil.
\end{align}
For any $k=1,\ldots,K,l=0,\ldots,\lfloor L\rfloor$, the coefficients $b_{l}^{(k)}$ are positive and defined inductively. Explicitly, 
\begin{align}
      b_{l}^{(1)}=0,\quad\text{if $l$ is even},\quad b_{l}^{(1)}=\frac{2\binom{l-1}{(l-1)/2}}{\pi2^{l-1}l},\quad\text{if $l$ is odd},\quad b_{l}^{(k+1)}=\sum_{l'=0}^{l}b_{l'}^{(k)}b_{l-l'}^{(1)}, \quad \forall k\geq1. \label{app:eq:blk}
\end{align}
Moreover, overall weights of $S(\rho)_{est}$ is bounded as follows.
\begin{align}
  \sum_{l=0}^{\lfloor L\rfloor}\sum_{s=D_{l}}^{U_l}\sum_{k=1}^{K}\frac{b_{l}^{(k)}\binom{l}{s}}{2^lk}\in O\left(\log(K)\right).
\end{align}
\end{lemma}
\begin{proof}
To prove the result in the Lemma, we only focus on the Taylor series of the natural logarithm. By Lemma~\ref{le:series}, for $K\in\Theta(\log(1/\epsilon)/\log(1/(1-\Lambda))$,
\begin{align}
\left|-\ln(x)-\sum_{k=1}^{K}\frac{1}{k}(1-x)^k\right|\leq\frac{\epsilon}{4}.\label{eq:ln_truncation}
\end{align}
We use the method in Lemma 37 of \cite{van2017quantum} to construct the Fourier series. The process proceeds as three major steps: 1. Transform the Taylor series into a linear combination of cosines. 2. Truncate the series to obtain a high-precision approximation. 3. Substitute all cosines with complex exponents. 

\textbf{Step 1.} 
As all eigenvalues $\lambda$ of $\rho$ in $[0,1]$, the identity below holds.
\begin{align}
1-\lambda=\frac{\arcsin(\sin((1-\lambda)\pi/2))}{\pi/2}=\frac{\arcsin(\cos(\lambda\pi/2))}{\pi/2}.
\end{align}
Substitute $\lambda$ with $x$ in the L.H.S of Eq.~\eqref{eq:ln_truncation} and $1-\lambda$ with $\arcsin(\cos(\lambda\pi/2))/\pi/2$.
\begin{align}
    \left|-\ln(\lambda)-\sum_{k=1}^{K}\frac{1}{k}\left(\frac{\arcsin(\cos(\lambda\pi/2))}{\pi/2}\right)^k\right|\leq\frac{\epsilon}{4}.
\end{align}
Recall the Taylor series of $\arcsin(y)$, for all $y\in[-1,1]$,
\begin{align}
    \arcsin(y)=\sum_{l=0}^{\infty}\frac{(2l)!}{4^l(l!)^2(2l+1)}y^{2l+1}=y+\frac{y^3}{6}+\frac{3y^5}{40}+\ldots
\end{align}
Using the Taylor series of $\arcsin(y)$, we deduce the Taylor series of $(\arcsin(y)/\pi/2)^k$, given by
\begin{align}
    \left(\frac{\arcsin(y)}{\pi/2}\right)^k=\sum_{l=0}^{\infty}b_{l}^{(k)}y^l.
\end{align}
Particularly, the coefficients $b_l^{(k)}$ could be inductively computed. Specifically speaking,
\begin{align}
    \left(\frac{\arcsin(y)}{\pi/2}\right)^{k+1}=\left(\frac{\arcsin(y)}{\pi/2}\right)^{k}\left(\sum_{l=0}^{\infty}b_{l}^{(1)}y^l\right).
\end{align}
For $k=1$, all $b_{l}^{(1)}$ correspond to the coefficients of the Taylor series of $\arcsin$ times $\frac{2}{\pi}$. That is,
\begin{align}
    b_{2l}^{(1)}=0,\quad b_{2l+1}^{1}=\frac{2}{\pi}\frac{(2l)!}{4^l(l!)^2(2l+1)}=\frac{2/\pi}{4^l(2l+1)}\binom{2l}{l}.\label{eq:arcsin}
\end{align}
For general $k$, suppose all coefficients $b_{l}^{(k)}$ are obtained, then the coefficients $b_{l}^{(k+1)}$ are computed as follows:
\begin{align}
    b_{l}^{(k+1)}&=\sum_{l'=0}^{l}b_{l'}^{(k)}b_{l-l'}^{(1)}.\label{eq:induction}
\end{align}
So, we could inductively compute all coefficients $b_{l}^{(k)}$ using Eqs.~\eqref{eq:arcsin}-\eqref{eq:induction}.

Consequently, we have deducted a series to approximate $-\ln(\lambda)$.
\begin{align}
    \left|-\ln(\lambda)-\sum_{k=1}^{K}\frac{1}{k}\sum_{l=0}^{\infty}b_{l}^{(k)}\cos^l(\lambda\pi/2)\right|\leq\frac{\epsilon}{4}. \label{eq:derived_series}
\end{align}

\textbf{Step 2.} Now, we truncate the series in Eq.~\eqref{eq:derived_series} to obtain a high-precision approximation series. 

First, we truncate the infinity series at order $L=\ln\left(\frac{4\sum_{k=1}^{K}1/k}{\epsilon}\right)\frac{1}{\Lambda^2}$. Next, we show the truncation error.
\begin{align}
    \left|\sum_{l=\lceil L\rceil}^{\infty}b_{l}^{(k)}\cos^l(\lambda\pi/2)\right|&\leq \sum_{l=\lceil L\rceil}^{\infty} b_{l}^{(k)}\left|\cos^l(\lambda\pi/2)\right|\\
    &\leq \cos^{L}(\lambda\pi/2)\sum_{l=\lceil L\rceil}^{\infty} b_{l}^{(k)}\\
    &\leq \cos^{L}(\lambda\pi/2)\\
    &= \sin^{L}((1-\lambda)\pi/2)\\
    &\leq (1-\lambda^2)^L\\
    &\leq e^{-\lambda^2L}\\
    &\leq e^{-\Lambda^2L}\\
    &\leq \frac{\epsilon}{4\sum_{k=1}^{K}1/k}.
\end{align}
Here, we have used facts that $b_{l}^{(k)}\geq 0$, $\sum_{l=1}^{\infty}b_l^{(k)}=1$, $\sin((1-\delta)\pi/2)\leq 1-\delta^2$ for all $\delta\in(0,1)$, $(1-\delta)\leq e^{-\delta}$, and $\lambda\in[\Lambda,1]$.

As shown above, the truncated series could act well as an approximation.
\begin{align}
    \left|-\ln(\lambda)-\sum_{k=1}^{K}\frac{1}{k}\sum_{l=0}^{\lfloor L\rfloor}b_{l}^{(k)}\cos^l(\lambda\pi/2)\right|\leq \frac{\epsilon}{2}. \label{eq:cos_series}
\end{align}

\textbf{Step 3.}
Now, we use the equality $\cos(z)=\frac{e^{iz}+e^{-iz}}{2}$ to rewrite the series in Eq.~\eqref{eq:cos_series}. 
\begin{align}
    \sum_{k=1}^{K}\frac{1}{k}\sum_{l=0}^{\lfloor L\rfloor}b_{l}^{(k)}\left[\frac{e^{i\lambda\pi/2}+e^{-i\lambda\pi/2}}{2}\right]^l=\sum_{k=1}^{K}\frac{1}{k}\sum_{l=0}^{\lfloor L\rfloor}b_{l}^{(k)}2^{-l}\sum_{s=0}^{l}\binom{l}{s}e^{i(2s-l)\lambda\pi/2}.
\end{align}
Particularly, this series could be further truncated by using the property of binomial distribution. By Chernoff's inequality, we have
\begin{align}
    \sum_{s=\lceil l/2\rceil+M_l}^{l}2^{-l}\binom{l}{s}\leq e^{-\frac{2M_l^2}{l}}.
\end{align}
Setting $M_l=\left\lceil\sqrt{\ln\left(\frac{4\sum_{k=1}^{K}1/k}{\epsilon}\right)\frac{l}{2}}\right\rceil$, suppose that $M_{l}\leq \lfloor l/2\rfloor $, then we could find that
\begin{align}
    \sum_{s=0}^{\lfloor l/2\rfloor-M_l}2^{-l}\binom{l}{s}=\sum_{s=\lceil l/2\rceil+M_l}^{l}2^{-l}\binom{l}{s}\leq e^{-\frac{2M_l^2}{l}} \leq \frac{\epsilon}{4\sum_{k=1}^{K}1/k}.
\end{align}

Eventually, we have obtained the desired series, shown below, up to precision $\epsilon$.
\begin{align}
    \sum_{l=0}^{\lfloor L\rfloor}\sum_{s=\lfloor l/2\rfloor-M_l}^{\lceil l/2\rceil+M_l}\left(\sum_{k=1}^{K}\frac{b_{l}^{(k)}}{k}\right)2^{-l}\binom{l}{s}e^{i(2s-l)\lambda\pi/2}.
\end{align}
The claimed result follows immediately from the formula of entropy.

In addition, the sum of all coefficients is bounded. This is because
\begin{align}
    \sum_{l=0}^{\lfloor L\rfloor}\sum_{s=\lfloor l/2\rfloor-M_l}^{\lceil l/2\rceil+M_l}\left(\sum_{k=1}^{K}\frac{b_{l}^{(k)}}{k}\right)2^{-l}\binom{l}{s}&\leq \sum_{l=0}^{\lfloor L\rfloor}\left(\sum_{k=1}^{K}\frac{b_{l}^{(k)}}{k}\right)\\
    &\leq \sum_{k=1}^{K}\frac{1}{k}\\
    &\leq \ln(K+1)+1.
\end{align}
Here we have used facts that $2^{-l}\sum_{s=0}^{l}\binom{l}{s}=1$, $\sum_{l=0}^{\infty}b_{l}^{(k)}=\left(\frac{\arcsin(1)}{\pi/2}\right)^k=1$, and the result in Lemma~\ref{le:series}.
\end{proof}

\subsection{Proof for Proposition~\ref{le:bound_binom}}\label{app:proposition1}
\begin{proposition}
For any constant $\beta\in(-1,0)\cup(0,+\infty)$, there exists a bound on the generalized binomial coefficient $\binom{\beta}{k}$. 
\begin{enumerate}
    \item For $\beta\in(-1,0)$ and any integer $k\geq 1$, $|\binom{\beta}{k}|\leq |\beta|$. 
    \item For $\beta\in(0,1]$ and any $k\geq 2$, $\left|\binom{\beta}{k}\right|\leq \left[1+\frac{\beta\ln\frac{(k+1)}{k^2}+\beta-1}{k}\right]^k.$
Particularly, $|\binom{\beta}{k}|\leq \frac{1}{e}$, if $k\in\Omega(1)$. 
\item For $\beta\in(1,+\infty)$ and $k\geq \beta+1$, $\left|\binom{\beta}{k}\right|\leq \left[1+\frac{\beta\ln\frac{(\beta+1)^2}{k}+2}{k}\right]^k.$
Particularly, $|\binom{\beta}{k}|\leq 1$, if $k\in \Omega((\beta+1)^2)$.
\end{enumerate}
Moreover, for an integer $K$, the sum $\sum_{k=1}^{K}\left|\binom{\beta}{k}\right|$ is bounded. 
\begin{align}
    \sum_{k=1}^{K}\left|\binom{\beta}{k}\right|\leq\left\{
    \begin{array}{ll}
     O\left( K+e^2(\beta+1)^{2\beta}\cdot\left[\ln(\lceil e^{\frac{2}{\beta}}(\beta+1)^2\rceil+1)+1\right]-e^{\frac{2}{\beta}}(\beta+1)^2 \right),  & \text{if $\beta\in(1,+\infty)$}, \\
        O(K), & \text{if $\beta\in(0,1]$},\\
     O\left(   |\beta| K \right), & \text{if $\beta\in(-1,0)$}.
    \end{array}\right.
\end{align}
\end{proposition}

\begin{proof}
Case 1: $\beta\in(1,+\infty)$. 

We consider an integer $k$ such that $k-1\geq\beta$. We use $[\beta]$ to denote the maximal integer that is less than $\beta$.
\begin{align}
    \left|\binom{\beta}{k}\right|&=\frac{\beta|\beta-1|\cdot|\beta-2|\cdot\ldots\cdot|\beta-k+1|}{k!} \\
    &=\frac{\beta}{k}\cdot\frac{|\beta-1|}{1}\cdot\frac{|\beta-2|}{2}\cdot\ldots\frac{|\beta-k+1|}{k-1}\cdot\\
&\leq\left[\frac{\frac{\beta}{k}+|\frac{\beta}{1}-1|+\ldots+|\frac{\beta}{[\beta]}-1|+\ldots+|\frac{\beta}{k-1}-1|}{k}\right]^k\\
&=\left[\frac{\frac{\beta}{k}+(\frac{\beta}{1}-1)+\ldots+(\frac{\beta}{[\beta]}-1)+\ldots+(1-\frac{\beta}{k-1})}{k}\right]^k\\
&=\left[\frac{\beta(1+\frac{1}{2}+\ldots+\frac{1}{[\beta]}+\frac{1}{k})-[\beta]+(k-1-[\beta])-\beta(\frac{1}{[\beta]+1}+\ldots+\frac{1}{k-1})}{k}\right]^k\\
&=\left[\frac{2\beta(1+\frac{1}{2}+\ldots+\frac{1}{[\beta]})+\frac{\beta}{k}-[\beta]+(k-1-[\beta])-\beta(1+\frac{1}{2}+\ldots+\frac{1}{k-1})}{k}\right]^k.
\end{align}
In the first inequality, we use the inequality of arithmetic and geometric means. 

Notice that sum $\sum_{l=1}^{k}\frac{1}{k}$ is a Harmonic number. By the properties of Harmonic numbers, we have 
\begin{align}
&1+\frac{1}{2}+\ldots+\frac{1}{[\beta]}\leq \ln([\beta]+1)+1\leq\ln(\beta+1)+1,\\
&1+\frac{1}{2}+\ldots+\frac{1}{k-1} \geq \ln(k).
\end{align}
With these properties, we derive a new bound on the generalized binomial coefficient.
\begin{align}
    \left|\binom{\beta}{k}\right|&\leq \left[\frac{2\beta(\ln(\beta+1)+1)+\frac{\beta}{k}-2[\beta]+k-1-\beta \ln(k)}{k}\right]^k\\
    &\leq\left[1+\frac{2\beta(\ln(\beta+1)+1)-2(\beta-1)-\beta\ln(k)}{k}\right]^k\\
    &=\left[1+\frac{\beta\ln\frac{(\beta+1)^2}{k}+2}{k}\right]^k.
\end{align}

Let $k\geq e^{\frac{2}{\beta}}(\beta+1)^2$, then we have $\beta\ln\frac{(\beta+1)^2}{k}+2\leq0$. Then $|\binom{\beta}{k}|$ will decrease exponentially fast as $k$ increases. As a result, it suffices to bound the generalized binomial coefficient using $k=\lceil e^{\frac{2}{\beta}}(\beta+1)^2\rceil$, i.e., $|\binom{\beta}{k}|\leq 1$.

For an integer $K\geq \lceil e^{\frac{2}{\beta}}(\beta+1)^2\rceil$, we can derive a bound on the sum $\sum_{k=1}^{K}|\binom{\beta}{k}|$. Note that for $k\geq  \lceil e^{\frac{2}{\beta}}(\beta+1)^2\rceil$, the generalized coefficient $\binom{\beta}{k}$ has a bound $1$.
\begin{align}
    \sum_{k=1}^{K}\left|\binom{\beta}{k}\right|& \leq \sum_{k=1}^{\lceil e^{\frac{2}{\beta}}(\beta+1)^2\rceil}\left[1+\frac{\beta\ln\frac{(\beta+1)^2}{k}+2}{k}\right]^k+K-\lceil e^{\frac{2}{\beta}}(\beta+1)^2\rceil\\
    &\leq \sum_{k=1}^{\lceil e^{\frac{2}{\beta}}(\beta+1)^2\rceil} \exp\left(\beta\ln\frac{(\beta+1)^2}{k}+2\right)+K-\lceil e^{\frac{2}{\beta}}(\beta+1)^2\rceil\\
    &=\sum_{k=1}^{\lceil e^{\frac{2}{\beta}}(\beta+1)^2\rceil}e^2\cdot \left[\frac{(\beta+1)^2}{k}\right]^\beta+K-\lceil e^{\frac{2}{\beta}}(\beta+1)^2\rceil.
\end{align}
In the second inequality, we use the inequality $\exp(x)\geq 1+x$ for all $x\geq0$.

Next, we give a bound.
\begin{align}
    \sum_{k=1}^{\lceil e^{\frac{2}{\beta}}(\beta+1)^2\rceil}e^2\cdot \left[\frac{(\beta+1)^2}{k}\right]^\beta&=e^2\cdot (\beta+1)^{2\beta}\sum_{k=1}^{\lceil e^{\frac{2}{\beta}}(\beta+1)^2\rceil}k^{-\beta}\\
    &<e^2\cdot (\beta+1)^{2\beta}\sum_{k=1}^{\lceil e^{\frac{2}{\beta}}(\beta+1)^2\rceil}\frac{1}{k}\\
    &\leq e^2\cdot (\beta+1)^{2\beta}\left[\ln(\lceil e^{\frac{2}{\beta}}(\beta+1)^2\rceil+1)+1\right].
\end{align}
In the first inequality, we use the fact that $k^{-\beta}<k^{-1}$ as $\beta>1$. 

Consequently, we readily obtain a bound for the sum.
\begin{align}
    \sum_{k=1}^{K}\left|\binom{\beta}{k}\right|\leq e^2\cdot (\beta+1)^{2\beta}\cdot\left[\ln(\lceil e^{\frac{2}{\beta}}(\beta+1)^2\rceil+1)+1\right]+K-e^{\frac{2}{\beta}}(\beta+1)^2.
\end{align}

Case 2: $\beta\in(0,1]$. 

Recall the form of the generalized binomial coefficient. 
\begin{align}
    \left|\binom{\beta}{k}\right|&=\frac{\beta(1-\beta)\ldots(k-1-\beta)}{k!}\\
    &=\frac{\beta}{k}\cdot(1-\beta)\ldots(1-\frac{\beta}{k-1})\\
    &\leq \left[\frac{\frac{\beta}{k}+(1-\beta)+\ldots+(1-\frac{\beta}{k-1})}{k}\right]^k\\
    &=\left[\frac{\frac{\beta}{k}+(k-1)-\beta(1+\frac{1}{2}+\ldots+\frac{1}{k-1})}{k}\right]^k\\
    &\leq \left[\frac{\frac{\beta}{k}+(k-1)-\beta\ln(k)}{k}\right]^k\\
    &=\left[1+\frac{\frac{\beta}{k}-1-\beta\ln(k)}{k}\right]^k\\
    &\leq \left[1+\frac{\beta[\ln(k+1)+1-\ln(k)]-\beta\ln(k)-1}{k}\right]^k\\
    &=\left[1+\frac{\beta\ln\frac{(k+1)}{k^2}+\beta-1}{k}\right]^k.
\end{align}
The first inequality is due to the inequality of arithmetic and geometric means. The second and third are the results of applying the properties of the Harmonic series.

Notice that for any integer $k\geq 4$, the numerator $\beta\ln\frac{(k+1)}{k^2}+\beta-1<\beta-1<-1$. Hence, we can derive a bound by setting $k\geq 4$.
\begin{align}
    \left|\binom{\beta}{k}\right|\leq \left[1+\frac{\beta\ln\frac{(k+1)}{k^2}+\beta-1}{k}\right]^k\leq \frac{1}{e}.
\end{align}
Immediately, the bound on the sum is given as $O(K)$.

Case 3: $\beta\in(-1,0)$.

For any integer $k$, the bound on the generalized coefficient is given as follows.
\begin{align}
    \left|\binom{\beta}{k}\right|&=\frac{(-\beta)(1-\beta)\ldots(k-1-\beta)}{k!}\\
    &\leq \frac{(-\beta)(2)\ldots(k)}{k!}\\
    &=-\beta.
\end{align}
Again, the bound on the sum is $O(|\beta|K)$.
\end{proof}

\subsection{Proof for Proposition~\ref{le:taylor_series}}\label{app:proposition2}
\begin{proposition}
For any constant $\alpha\in(0,1)\cup(1,+\infty)$ and $\xi\in(0,1)$, there exists an integer $K$ such that, for any quantum state $\rho$ with eigenvalue lower bound $\Lambda$,
\begin{align}
    \left|\tr(\rho^\alpha)-1-\sum_{k=1}^{K}\binom{\beta}{k}\tr\left(\rho(\rho-I)^k\right)\right|\leq\xi.\label{app:eq:taylor_approximation}
\end{align}
Particular, the choice of integer $K$ is shown below.
\begin{align}
K\in\left\{
    \begin{array}{ll}
        \max\{\Omega(\alpha^2), \Omega(\log(\Lambda\xi)/\log(1-\Lambda))\} & \text{if $\alpha\in(2,+\infty)$}, \\
        \Omega(\log(\Lambda\xi)/\log(1-\Lambda)) &  \text{if $\alpha\in(0,1)\cup(1,2]$}.
    \end{array}\right.
\end{align}
\end{proposition}
\begin{proof}
Let $\rho=\sum_{j}p_j \op{\mathbf{e}_j}{\mathbf{e}_j}$ be the Schmidt decomposition of $\rho$, and expand $\rho^\beta$ into the Taylor series.
\begin{align}
    \left\|\rho^\alpha-\rho-\sum_{k=1}^{K}\binom{\beta}{k}\rho(\rho-I_{\rm supp})^k\right\|&=\left\|\sum_{k=K+1}^{\infty}\binom{\beta}{k}\sum_{j}p_j(p_j-1)^k\op{\mathbf{e}_j}{\mathbf{e}_j}\right\|\\
    &\leq \sum_{k=K+1}^{\infty}\left|\binom{\beta}{k}\right|(1-\Lambda)^k\\
    &\leq \frac{(1-\Lambda)^{K+1}}{\Lambda}.
\end{align}
By Proposition \ref{le:bound_binom}, we can find an integer $K$ such that $|\binom{\beta}{k}|\leq 1$ for all $k\geq K$.

Furthermore, letting $K\geq \log(\Lambda\xi)/\log(1-\xi)$, we can verify that $(1-\Lambda)^{K+1}/\Lambda \leq \xi$.
\end{proof}

\subsection{Proof for Lemma~\ref{le:renyi_fourier}}\label{app:renyi_fourier}
\begin{lemma}
Consider a quantum state $\rho\in\mathbb{C}^{2^n\times 2^n}$. Let $\Lambda\in(0,1)$ be a lower bound on all non-zero eigenvalues. For any $\alpha\in(0,1)\cup(1,+\infty)$, there exists an estimate $R_{\alpha}(\rho)_{est}$ $\alpha$-R\'enyi entropy $R_{\alpha}(\rho)$ up to precision $\epsilon$. To be specific,
\begin{align}
    R_{\alpha}(\rho)_{est}&=\frac{1}{1-\alpha}\log F_{\alpha}(\rho),
\end{align}
where $F_{\alpha}(\rho)$ satisfies $|F_{\alpha}(\rho)-\tr(\rho^\alpha)|\leq \xi$.
\begin{align}
  F_{\alpha}(\rho)= 1+ \sum_{l=0}^{\lfloor L\rfloor}\sum_{s=D_l}^{U_l}\left(\sum_{k=1}^{K}(-1)^kb_{l}^{(k)}\binom{\alpha-1}{k}\right)2^{-l}\binom{l}{s}\tr(\rho\cdot \cos(\rho t(s,l))). \label{app:eq:renyi}
\end{align}
In particular, the relation between $\epsilon$ and $\xi$ is given in Eq.~\eqref{eq:renyi_precision}, and definition of all $b_{l}^{(k)}$ are given in Eq.~\eqref{eq:blk}.
And the parameters of $F_{\alpha}(\rho)$ are given as follows. Coefficient $t(s,l)=\frac{(2s-l)\pi}{2}$, and $U_{l}=\min\{l,\lceil \frac{l}{2}\rceil+M_l\}$ and $D_l=\max\{0,\lfloor \frac{l}{2}\rfloor-M_l\}$. Moreover,
\begin{align}
    K=\Theta\left(\frac{\log(\Lambda\xi)}{\log(1-\Lambda)}+\alpha^2\right),\quad L=\ln\left(\frac{4\sum_{k=1}^{K}|\binom{\alpha-1}{k}|}{\xi}\right)\frac{1}{\Lambda^2},\quad M_l=\left\lceil\sqrt{\ln\left(\frac{4\sum_{k=1}^{K}|\binom{\alpha-1}{k}|}{\xi}\right)\frac{l}{2}}\right\rceil.
\end{align}
And the overall weights of $F_{\alpha}(\rho)$ are bounded by $\sum_{k=1}^{K}|\binom{\alpha-1}{k}|$, and bounds on $\sum_{k=1}^{K}|\binom{\alpha-1}{k}|$ are given in Table~\ref{tab:my_label_weights}.

\end{lemma}
\begin{proof}
We focus on constructing a Fourier series approximation to $\rho^{\beta}$, where $\beta=\alpha-1$. 

By Lemma \ref{le:taylor_series}, we can find a Taylor series in Eq.~\eqref{eq:taylor_approximation} that approximates to $\rho^{\beta}$ with the error $(1-\Lambda)^{K+1}/\Lambda$. Setting $K=\Theta\left(\frac{\log(\Lambda\xi)}{\log(1-\Lambda)}+\alpha^2\right)$, the approximation error is suppressed to $\xi/4$. 

We show the process of constructing the Fourier series in three steps.

\textbf{Step 1}. For simplicity, we only consider the eigenvalue $\lambda$ instead of state $\rho$. Notice that 
\begin{align}
    \lambda-1=\frac{\arcsin(\sin(\frac{(\lambda-1)\pi}{2}))}{\frac{\pi}{2}}=-\frac{\arcsin(\cos(\lambda\pi/2))}{\pi/2}. \label{eq:sin}
\end{align}
Take the relation in Eq.~\eqref{eq:sin} into the series in Eq.~\eqref{eq:taylor_approximation}, then we can find an approximation to $\tr(\rho^\alpha)$.
\begin{align}
    \tr(\rho^\alpha)=\tr(\rho\cdot\rho^\beta)\approx 1+\sum_{k=1}^{K}\binom{\beta}{k}(-1)^k\tr\left(\rho\left(\frac{\arcsin(\cos(\rho\pi/2))}{\pi/2}\right)^k\right).
\end{align}

Next, expand the function $(\arcsin(y)/\pi/2)^k$ into the Taylor series, which is given in the following formula.
\begin{align}
    (\arcsin(y)/\pi/2)^k=\sum_{l=0}^{\infty}b_{l}^{(k)}y^l.
\end{align}
The coefficients $b_{l}^{(k)}$ can be efficiently calculated.

Then, $\tr(\rho^\alpha)$ could be written as follows.
\begin{align}
\tr(\rho^{\alpha})&\approx1+\sum_{k=1}^{K}\binom{\beta}{k}(-1)^k\sum_{l=0}^{\infty}b_{l}^{(k)}\tr\left(\rho\cdot \cos(\rho\pi/2)^l\right)\\
&=1+\sum_{k=1}^{K}\binom{\beta}{k}(-1)^k\sum_{l=0}^{\infty}b_{l}^{(k)}\tr\left(\rho\cdot \left(\frac{e^{-i\rho \pi/2}+e^{i\rho \pi/2}}{2}\right)^l\right)\\
&=1+\sum_{k=1}^{K}\binom{\beta}{k}(-1)^k\sum_{l=0}^{\infty}b_{l}^{(k)}\sum_{s=0}^{l}\frac{\binom{l}{s}}{2^l}\tr\left(\rho\cdot e^{i\rho(2s-l)\pi/2}\right).
\end{align}
Consequently, we could truncate the above series to obtain an approximation of $\tr(\rho^\alpha)$ at order $L$.
\begin{align}
    \tr(\rho^\alpha)\approx 1+\sum_{k=1}^{K}\binom{\beta}{k}(-1)^k\sum_{l=0}^{\lfloor L \rfloor}b_{l}^{(k)}\sum_{s=0}^{l}\frac{\binom{l}{s}}{2^l}\tr\left(\rho\cdot \cos\left(\rho\frac{(2s-l)\pi}{2}\right)\right). \label{eq:power_approximation}
\end{align}
By setting $L=\ln\left(\frac{4\sum_{k=1}^{K}|\binom{\beta}{k}|}{\xi}\right)\frac{1}{\Lambda^2}$, we would obtain a series approximating to $\tr(\rho^{\alpha})$ up to precision $\xi/2$. This is because, for any non-zero eigenvalue $\lambda$, we have
\begin{align}
    \left|\sum_{l=\lceil L\rceil}^{\infty}b_{l}^{(k)}\cos^l(\lambda\pi/2)\right|&\leq \sum_{l=\lceil L\rceil}^{\infty} b_{l}^{(k)}\left|\cos^l(\lambda\pi/2)\right|\\
    &\leq \cos^{L}(\lambda\pi/2)\sum_{l=\lceil L\rceil}^{\infty} b_{l}^{(k)}\\
    &\leq \cos^{L}(\lambda\pi/2)\\
    &= \sin^{L}((1-\lambda)\pi/2)\\
    &\leq (1-\lambda^2)^L\\
    &\leq e^{-\lambda^2L}\\
    &\leq e^{-\Lambda^2L}\\
    &\leq \frac{\xi}{4\sum_{k=1}^{K}|\binom{\beta}{k}|}.
\end{align}
Here, we have used facts that $b_{l}^{(k)}\geq 0$, $\sum_{l=1}^{\infty}b_l^{(k)}=1$, $\sin((1-\delta)\pi/2)\leq 1-\delta^2$ for all $\delta\in(0,1)$, $(1-\delta)\leq e^{-\delta}$, and $\lambda\in[\Lambda,1]$.

\textbf{Step 3.}
Now, we use the equality $\cos(z)=\frac{e^{iz}+e^{-iz}}{2}$ to rewrite the series in Eq.~\eqref{eq:power_approximation}. 
\begin{align}
    \sum_{k=1}^{K}(-1)^k\binom{\beta}{k}\sum_{l=0}^{\lfloor L\rfloor}b_{l}^{(k)}\left[\frac{e^{i\lambda\pi/2}+e^{-i\lambda\pi/2}}{2}\right]^l=\sum_{k=1}^{K}(-1)^k\binom{\beta}{k}\sum_{l=0}^{\lfloor L\rfloor}b_{l}^{(k)}2^{-l}\sum_{s=0}^{l}\binom{l}{s}e^{i(2s-l)\lambda\pi/2}.
\end{align}
Particularly, this series could be further truncated by using the property of binomial distribution. By Chernoff's inequality, we have
\begin{align}
    \sum_{s=\lceil l/2\rceil+M_l}^{l}2^{-l}\binom{l}{s}\leq e^{-\frac{2M_l^2}{l}}.
\end{align}
Setting $M_l=\left\lceil\sqrt{\ln\left(\frac{4\sum_{k=1}^{K}|\binom{\beta}{k}|}{\epsilon}\right)\frac{l}{2}}\right\rceil$, suppose that $M_{l}\leq \lfloor l/2\rfloor $, then we could find that
\begin{align}
    \sum_{s=0}^{\lfloor l/2\rfloor-M_l}2^{-l}\binom{l}{s}=\sum_{s=\lceil l/2\rceil+M_l}^{l}2^{-l}\binom{l}{s}\leq e^{-\frac{2M_l^2}{l}} \leq \frac{\xi}{4\sum_{k=1}^{K}|\binom{\beta}{k}|}.
\end{align}

Eventually, we have obtained the desired series $F_{\alpha}(\rho)$, shown below, up to precision $\xi$.
\begin{align}
  F_{\alpha}(\rho)= 1+ \sum_{l=0}^{\lfloor L\rfloor}\sum_{s=\lfloor l/2\rfloor-M_l}^{\lceil l/2\rceil+M_l}\left(\sum_{k=1}^{K}(-1)^kb_{l}^{(k)}\binom{\beta}{k}\right)2^{-l}\binom{l}{s}e^{i(2s-l)\lambda\pi/2}.
\end{align}
The claimed result follows immediately from the formula of entropy.

In addition, the sum of all coefficients is bounded. This is because
\begin{align}
    \sum_{l=0}^{\lfloor L\rfloor}\sum_{s=\lfloor l/2\rfloor-M_l}^{\lceil l/2\rceil+M_l}\sum_{k=1}
    ^{K}\left|(-1)^kb_{l}^{(k)}\binom{\beta}{k}2^{-l}\binom{l}{s}\right|\leq \sum_{l=0}^{\lfloor L\rfloor}\left|\sum_{k=1}^{K}b_{l}^{(k)}\binom{\beta}{k}\right|\leq \sum_{k=1}^{K}\left|\binom{\beta}{k}\right|.
\end{align}
Here we have used facts that $2^{-l}\sum_{s=0}^{l}\binom{l}{s}=1$, $\sum_{l=0}^{\infty}b_{l}^{(k)}=\left(\frac{\arcsin(1)}{\pi/2}\right)^k=1$. 

By Proposition~\ref{le:bound_binom}, we can find the upper bound on the sum of generalized coefficients as claimed.

\end{proof}

\section{Quantum circuits}
\subsection{Proof for Eq.~(\ref{eq:difference})} \label{sec:exponent_inequality}
\renewcommand{\thelemma}{S\arabic{lemma}}
\setcounter{lemma}{1}
\begin{lemma}\label{le:exponent_inequality}
For any quantum state $\rho$, time $t>0$, and integer $K>1$, the trace norm of $\sum_{k\geq K}\frac{(-i\rho t)^k}{k!}$ is bounded by $\sqrt{2}t^K/K!$.
\end{lemma}
\begin{proof}
Suppose $\rho$ has a spectral decomposition
\begin{align}
    \rho=\sum_{j}\lambda_j\op{\mathbf{e}_j}{\mathbf{e}_j}.
\end{align}
As the trace norm is the sum of all singular values, we could deduce that
\begin{align}
    \left\|\sum_{k\geq K}\frac{(-i\rho t)^k}{k!}\right\|_{tr}&=\left\|\sum_{j}\sum_{k\geq K}\frac{(-i\lambda_j t)^k}{k!}\op{\mathbf{e}_j}{\mathbf{e}_j}\right\|_{tr}\\
    &=\sum_{j}\left|\sum_{k\geq K}\frac{(-i\lambda_j t)^k}{k!}\right|\\
    &\leq\sum_{j}\frac{\sqrt{2}(\lambda_j t)^{K}}{K!}\\
    &\leq\sum_{j}\frac{\sqrt{2}\lambda_j( t)^{K}}{K!}\\
    &\leq \frac{\sqrt{2}t^K}{K!}.
\end{align}
Here, we have used facts that $\left| \sum_{k\geq K}\frac{(-i\lambda_j t)^k}{k!}\right|\leq \frac{\sqrt{2}(\lambda_j t)^K}{K!}$ for all $\lambda_j$, and $\sum_{j}\lambda_j=1$.


To complete the proof, we show that, for any $x\in\mathbb{R}$ and integer $K>1$,
\begin{align}
    \left|e^{-ix}-\sum_{k=0}^{K-1}\frac{(-ix)^k}{k!}\right|\leq \frac{\sqrt{2}|x|^{K}}{K!}.
\end{align}
Given arbitrary integer $K>0$, we use Euler's formula and expand the triangle functions to the Taylor series. The error is given by
\begin{align}
   e^{ix}- \sum_{k=1}^{K}\frac{(ix)^k}{k!}&=\cos(x)+i\sin(x)-\left[\sum_{p} (i)^{2p+1}\frac{x^{2p+1}}{(2p+1)!}+\sum_{p} (i)^{2p}\frac{x^{2p}}{(2p)!}\right]\\
    &=\cos(x)+i\sin(x)-\left[i\sum_{p} (-1)^{p}\frac{x^{2p+1}}{(2p+1)!}+\sum_{p} (-1)^{p}\frac{x^{2p}}{(2p)!}\right]\\
    &=\left[\cos(x)-\sum_{p} (-1)^{p}\frac{x^{2p}}{(2p)!}\right]+i\left[\sin(x)-\sum_{p} (-1)^{p}\frac{x^{2p+1}}{(2p+1)!}\right],
\end{align}
where $1\leq p\leq [K/2]$.

Notice that $\sum_{p} (-1)^{p}\frac{x^{2p+1}}{(2p+1)!}$ and $\sum_{p} (-1)^{p}\frac{x^{2p}}{(2p)!}$ are truncated Taylor series of $\sin(x)$ and $\cos(x)$ up to order $K$, respectively. By the Taylor's theorem, we have
\begin{align}
    &\cos(x)-\sum_{p} (-1)^{p}\frac{x^{2p}}{(2p)!}= \frac{\cos^{(K+1)}(\zeta)}{(K+1)!}x^{K+1},\\
    &\sin(x)-\sum_{p} (-1)^{p}\frac{x^{2p+1}}{(2p+1)!}= \frac{\sin^{(K+1)}(\xi)}{(K+1)!}x^{K+1}.
\end{align}
where $\zeta$ and $\xi$ are some values between $0$ and $x$. Immediately, we have
\begin{align}
    &\left|\cos(x)-\sum_{p} (-1)^{p}\frac{x^{2p}}{(2p)!}\right|\leq \frac{|x|^{K+1}}{(K+1)!},\\
    &\left|\sin(x)-\sum_{p} (-1)^{p}\frac{x^{2p+1}}{(2p+1)!}\right|\leq \frac{|x|^{K+1}}{(K+1)!}.
\end{align}
Finally, the result immediately follows, and the proof is finished.
\end{proof}

\subsection{Proof for Proposition~\ref{le:swap}}\label{app:swap}
\setcounter{proposition}{4}
\begin{proposition}
For arbitrary time $\Delta t\in(-1,1)$, define two rotations $R_1$ and $R_2$ as in Eqs.~\eqref{eq:rot_1}-\eqref{eq:rot_2}. Define a circuit module $W$ as in Figure \ref{figure:circuit_W} and a unitary $A=-W(I_{2}-2P)W^{\dagger}(I_{2}-2P)W$, where $P=\op{00}{00}$, and $I_{2}$ denotes the identity acting on $\ket{00}$. Then the unitary $e^{-i\mathcal{S}\Delta t}$ can be simulated in the sense that
\begin{align}
    P\otimes e^{-i\mathcal{S}\Delta t}=PAP.
\end{align}
\end{proposition}
\begin{proof}
First, we can easily show that
\begin{align}
    PAP&=3PWP-4PWPW^\dagger PWP.
\end{align}
Particularly, an important property of $W$ is 
\begin{align}
    \bra{00}W\ket{00}&=\frac{1}{2} \left(\cos(\Delta t)I_{2n}-i\sin(\Delta t)\mathcal{S}\right)=\frac{1}{2} e^{-i\mathcal{S}\Delta t}.
\end{align}
Then, 
\begin{align}
    &PWP=\frac{1}{2}P\otimes e^{-iS\Delta t},\\
    &PWPW^{\dagger}PWP=\frac{1}{4}PWP.
\end{align}
Last, the result immediately follows.

\end{proof}

\subsection{Gate decomposition}\label{sec:circuit_decomposition}
This section decomposes the quantum gates in Figure~\ref{figure:circuit_A} into primitive single/two-qubit gates. We primarily consider decomposing controlled gates c-${\rm select}(\mathcal{S})$, and anti-controlled gates oc-$R_2$. 

\paragraph{Circuit of oc-$R_2$.} Here, we let $R_{2}=R_{y}(\theta_2)$. By the relation that $XR_{y}(-\theta)X=R_{y}(\theta)$, we can decompose oc-$R_{y}(\theta_2)$ as shown in  Figure~\ref{figure:circuit_oc_R_2}.
\begin{figure}[htb]
\[ 
\Qcircuit @C=.7em @R=0.5em {
 & \gate{X}& \ctrl{1} & \qw & \ctrl{1} & \gate{X} & \qw  \\
  & \gate{R_{y}(\theta_2)} & \targ & \gate{R_{y}(-\theta_2/2)} & \targ & \gate{R_{y}(-\theta_2/2)} & \qw 
}
\]
\caption{\footnotesize Quantum circuit for anti-controlled rotation oc-$R_2$.}
\label{figure:circuit_oc_R_2}
\end{figure}
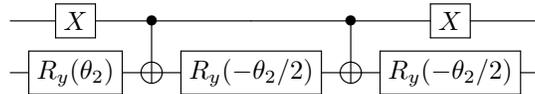


\paragraph{Circuit of ${\rm c-select}(\mathcal{S})$.} 
Notice that ${\rm select}(\mathcal{S})$ is a product of two operations: $P_1=(\op{0}{0}+(-i{\rm sgn}(t)\op{1}{1})\otimes I_{2n}$ and $P_2=\op{0}{0}\otimes I_{2n}+\op{1}{1}\otimes \mathcal{S}$. 
Hence, the c-${\rm select}(\mathcal{S})$ is to separately apply c-$P_1/P_2$. One key component of c-$P_1$ is the controlled phase gate c-$S$, which consists of $T$ gate, CNOT, and $R_z$. The decomposition is depicted in Figure~\ref{figure:circuit_c_S}.
\begin{figure}[htb]
\[ 
\Qcircuit @C=.7em @R=0.5em {
    & \gate{T} & \ctrl{1} & \qw & \ctrl{1} & \qw  & \qw    \\
    & \gate{R_{z}(\pi/2)} & \targ & \gate{R_{z}(-\pi/4)} & \targ & \gate{R_{z}(-\pi/4)} & \qw 
}
\]
\caption{\footnotesize Quantum circuit for controlled phase gate c-$S$.}
\label{figure:circuit_c_S}
\end{figure}
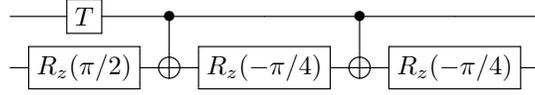

As for the c-$P_2$ gate, we first append one more measure register $\ket{0}$ and then use Toffoli and CNOT gates. Note that the decomposition of Toffoli gate can be found in \cite{nielsen2002quantum}. In consequence, our circuit is composed of primitive single/two-qubit gates. 

Ultimately, combining the circuits of c-$P_1/P_2$ leads to the circuit for c-${\rm select}(\mathcal{S})$, which is depicted in Figure~\ref{figure:circuit_c_swap}.
\begin{figure}[htb]
\[ 
\Qcircuit @C=1em @R=0.5em {
    &  \ctrl{1}       & \qw  & \qw        & \qw & \qw&\qw&\qw&\qw & \ctrl{1} & \qw & \ctrl{1} & \ctrl{1} & \qw\\
    &  \ctrl{1}       & \qw  & \qw        & \qw & \qw&\qw&\qw&\qw & \ctrl{1}  &\qw & \gate{S} & \gate{Z} & \qw \\
  \lstick{\op{0}{0}}  &  \targ       & \ctrl{2}  & \qw        & \qw & \ctrl{3}&\qw&\qw&\ctrl{4} & \targ  &\qw  & \qw & & \\
    &  &  &  & &  & & & &  &  & & & &\\
    &  \targ     & \ctrl{4}  & \targ      & \qw &\qw&\qw&\qw&\qw & \qw & \qw & \qw & &\\
  \lstick{\rho} &\qw&\qw&\qw&\targ&\ctrl{4}&\targ&\qw&\qw & \qw & \qw & \qw & &\\
    &\qw&\qw&\qw&\qw&\qw&\qw&\targ&\ctrl{4} & \targ& \qw & \qw &   & & \\
    &&&&&&&& & & &  & & \\
    & \ctrl{-4}  & \targ     & \ctrl{-4}  & \qw &\qw&\qw&\qw&\qw & \qw & \qw & \qw & & \\
  \lstick{\rho}&\qw&\qw&\qw&\ctrl{-4}&\targ&\ctrl{-4}&\qw&\qw & \qw & \qw & \qw & & \\
    &\qw&\qw&\qw&\qw&\qw&\qw&\ctrl{-4}&\targ & \ctrl{-4}& \qw  & \qw& &
}
\]
\caption{\footnotesize Quantum circuit for implementing controlled ${\rm select}(\mathcal{S})$. Here we take three-qubit state $\rho$ as example. The circuit appends one qubit $\ket{0}$. The decomposition of the c-$S$ is given in Figure~\ref{figure:circuit_c_S}. Particularly, the c-$Z$ gate is applied only when $t>0$.}
\label{figure:circuit_c_swap}
\end{figure}
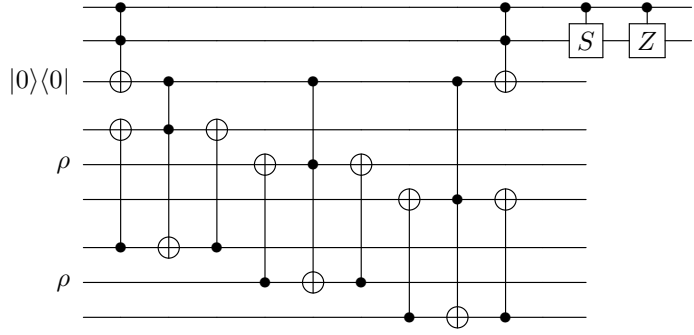

\section{Quantum algorithms}\label{app:renyi_algorithm}
\setcounter{theorem}{1}
\begin{theorem}
Consider a quantum state $\rho\in\mathbb{C}^{2^n\times2^n}$. Let $\Lambda$ be the lower bound on all non-zero eigenvalues of $\rho$. Suppose we have access to copies of $\rho$, then Algorithm~\ref{alg:renyi_entropy} outputs an estimate of $\alpha$-R\'enyi entropy $R_{\alpha}(\rho)$ up to precision $\epsilon$, succeeding with probability at least $1-\delta$. Furthermore, the total amount of the needed copies of $\rho$ and single/two-qubit gates, in the worst case, are shown in the Table~\ref{tab:renyi}.
\end{theorem}
\begin{proof}
\textbf{Correctness analysis}.

Recall the relation between $R_{\alpha}(\rho)$ and $\tr(\rho^\alpha)$. Thus, we focus on analyzing the estimation of $\tr(\rho^\alpha)$. For this purpose, we write the Fourier series $F_{\alpha}(\rho)$ as follows.
\begin{align}
    F_{\alpha}(\rho)=1+\sum_{l=0}^{\lfloor L\rfloor}\sum_{s=D_l}^{U_l}f(s,l)\tr\left(\rho \cos\left(\rho\frac{(2s-l)\pi}{2}\right)\right).
\end{align}
Coefficients $f(s,l)$ are given by
\begin{align}
    f(s,l)=\left(\sum_{k=1}^{K}b_{l}^{(k)}(-1)^k\binom{\beta}{k}\right)\frac{\binom{l}{s}}{2^{l}}, \quad\forall s,l. \label{eq:renyi_weights}
\end{align}
where $\beta=\alpha-1$. Let $\mathbf{f}$ be a vector that consists of $f(s,l)$. Then the $\ell_1$-norm of $\mathbf{f}$ is bounded, i.e., $\|\mathbf{f}\|_{\ell_1}\in O\left(\sum_{k=1}^{K}|\binom{\beta}{k}|\right)$. In addition, the bound on $\sum_{k=1}^{K}|\binom{\beta}{k}|$ can be found in Lemma~\ref{le:renyi_fourier}.

Next, define an importance sampling as follows:
\begin{align}
    &\mathbf{R}=\tr\left(\rho \cos\left(\rho\frac{(2s-l)\pi}{2}\right)\right) \quad \text{with prob. $\frac{|f(s,l)|}{ \|\mathbf{f}\|_{\ell_1}}$}.\label{eq:renyi_random_variable}
\end{align}
The random variable $\mathbf{R}$ indicates that each Fourier term associated with $(s,l)$ is sampled with probability proportional to its weight $f(s,l)$. Then, the Fourier series could be rewritten as an expectation of $\mathbf{R}$. 
\begin{align}
    F(\rho)=1+\left\|\mathbf{f}\right\|_{\ell_1}\cdot\mathbf{E}\left[\mathbf{R}\right].
\end{align}

\textbf{Cost analysis}.

Note that the sample mean could estimate the expectation. The estimation accuracy replies on the number of samples. By Chebyshev's inequality, $O({\bf Var}/\epsilon^2)$ samples are sufficient to derive an estimate of the expectation with precision $\epsilon$ and high probability, where ${\bf Var}$ denotes the variance, and $\epsilon$ is the precision. Meanwhile, the probability could be boosted to $1-\delta$ at the cost of an additional multiplicative factor $O(\log(1/\delta))$ according to Chernoff bounds. Alternatively, by Hoeffding's inequality, we only need $O(\log(1/\delta)/\epsilon^2)$ samples to derive an estimate with precision $\epsilon$ and probability larger than $1-\delta$.

Regarding random variable $\mathbf{R}$, the variance is less than $1$. We set the precision as $\varepsilon=\xi/\left\|\mathbf{f}\right\|_{\ell_1}$ and failure probability $\delta/2$. Then the number of required samples is 
\begin{align}
    \#\textrm{Samples}&= O\left(\frac{1}{\varepsilon^2}\log\left(\frac{2}{\delta}\right)\right)=O\left(\frac{\|\mathbf{f}\|_{\ell_1}^2}{\xi^2}\log\left(\frac{2}{\delta}\right)\right).\label{eq:renyi_samples}
\end{align}
It means that there are at most $\#{\rm Samples}$ terms needing to estimate via quantum circuits. 

On the other hand, notice that the Fourier series $F(\rho)$ consists of $N=\sum_{l=0}^{L}(2M_l+1)$ terms in all. Hence, it suffices to estimate each term with probability $1-\delta/2N$. In this way, the overall failure probability is at most $\delta$ by union bound.

When estimating the Fourier series, we need to measure the resultant state after evolving the input state $\rho$ by the circuits (e.g., please refer to Figure~\ref{figure:circuit_instance}). Note that the measurement outcome is evaluated nondeterministically, which implies that the estimation could fail. To suppress the failure probability to $\delta$, we suffice to evaluate each term with a failure probability at most $\delta/2N$. In consequence, for each term, the number of needed measurements is
\begin{align}
  \#{\rm Measurements}=  O\left(\frac{\log(2N/\delta)}{\varepsilon^2}\right).\label{eq:renyi_measurements}
\end{align}
Immediately, the total number of measurements is at most
\begin{align}
C_{m}&=\#{\rm Sample}\times \#{\rm Measurements}=O\left(\frac{\|\mathbf{f}\|_{\ell_1}^4}{\xi^4}\log\left(\frac{2}{\delta}\right)\log\left(\frac{N}{\delta}\right)\right).
\end{align}

Now, we consider the number of copies of state $\rho$. Note that running the circuit once costs $\widetilde{O}(t^2/\epsilon)$ by Proposition~\ref{th:main_cost}. Here, the largest time is $O(M_L)=O\left({\ln\left(\frac{4\sum_{k=1}^{K}|\binom{\beta}{k}|}{\xi}\right)\frac{1}{\Lambda}}\right)$. As shown above, we have to run circuits $C_m$ times in the entropy estimation, then the number of overall copies is at most
\begin{align}
     C_{\rho}=C_{m}\times \widetilde{O}\left(\ln^2\left(\frac{\sum_{k=1}^{K}|\binom{\beta}{k}|}{\epsilon}\right)\frac{1}{\varepsilon\Lambda^2}\right)=\widetilde{O}\left(\frac{\|\mathbf{f}\|_{\ell_1}^5}{\xi^5\Lambda^2}\right).
\end{align}

By Proposition~\ref{th:term}, for $\tr(\rho\cos(\rho t))$, the number of primitive single/two-qubit gates scales $O(nt^2/\epsilon)$. Here, the overall gate counts for the entropy estimation, in the worst case, is 
\begin{align}
    C_{g}=\# {\rm Sample} \times O(nM_L^2/\varepsilon)=\widetilde{O}\left(\frac{n\|\mathbf{f}\|_{\ell_1}^3}{\xi^3\Lambda^2}\right).
\end{align}

Finally, using the relation between $\xi$ and $\epsilon$ in Eq.~\eqref{eq:renyi_precision}, we can finish the proof for the claimed.

\end{proof}

\section{Numerical simulation}
\subsection{Quantum states used in Figure \ref{fig:histgram}}\label{sec:quantum state}

Given the condition that the eigenvalue of the minimum states is greater than lower bound $\Lambda=0.35$, the four states are generated randomly.

\begin{align*}
    \rho_1 = 
    \begin{pmatrix}
        0.37336237 & -0.02597119\\
        -0.02597119 & 0.62663763
    \end{pmatrix},
    \rho_2 =& 
    \begin{pmatrix}
        0.42050704 & -0.08174482\\
        -0.08174482 & 0.57949296
    \end{pmatrix},\\
    \rho_3 = 
    \begin{pmatrix}
        0.58221067 & -0.04587666\\
        -0.04587666 & 0.41778933
    \end{pmatrix},
    \rho_4 =& 
    \begin{pmatrix}
        0.42932114 & -0.02696812\\
        -0.02696812 & 0.57067886
    \end{pmatrix}.
\end{align*} 
\end{document}